\documentclass[12pt,journal,draftcls,letterpaper,onecolumn]{IEEEtran}

\usepackage{graphicx}  

\usepackage{subfigure} 

\usepackage{multirow}

\usepackage{amsmath}
\usepackage{amssymb}
\usepackage[justification=centering]{caption}

\usepackage{amsmath,bm}

\usepackage{algorithmic}
\usepackage{xcolor}
\usepackage{color}
\newtheorem{theorem}{Theorem}

\newtheorem{corollary}{Corollary}
\newtheorem{definition}{Definition}

\newtheorem{proposition}{Proposition}

\newtheorem{remark}{Remark}

\definecolor{darkgreen}{rgb}{0.1,0.5,.5}
\definecolor{darkred}{rgb}{0.8,0,0}
\definecolor{teal}{rgb}{0.05,0.32,0.41}
\definecolor{darkblue}{rgb}{0,0.0,0.5}
\definecolor{blackgreen}{rgb}{0,0.4,0}
\definecolor{purple}{rgb}{0.5,0,0.3}
\definecolor{grey}{rgb}{0.7,0.5,0.5}
\definecolor{orange}{rgb}{0.6,0.4,0.1}

\begin{document}
%
\title{Relay-Aided MIMO Cellular Networks Using Opposite Directional Interference Alignment}
%
%
\author{Hoyoun Kim and Jong-Seon No
\thanks{H. Kim and J.-S. No are with the Department of Electrical and Computer
Engineering, INMC, Seoul National University, Seoul 08826, Korea. e-mail:
ferui@ccl.snu.ac.kr, jsno@snu.ac.kr.}
}

%
%
%
\markboth{Submission for...}{Regular Paper}



\maketitle

\begin{abstract}

In this paper, we propose an interference alignment (IA) scheme for the multiple-input multiple-output (MIMO) uplink cellular network with the help of a relay which operates in half-duplex mode. 
The proposed scheme only requires global channel state information (CSI) knowledge at the relay, with no transmitter beamforming and time extension at the user equipment (UE), which differs from conventional IA schemes for cellular networks.
We derive the feasibility condition of the proposed scheme for the general cellular network configuration and analyze the degrees-of-freedom (DoF) performance of the proposed IA scheme while providing a closed-form beamformer design at the relay. 
Extensions of the proposed scheme to downlink and full-duplex cellular networks are also proposed in this paper. 
The DoF performance of the proposed schemes is compared to that of a linear IA scheme for a cellular network with no time extension. 
It is also shown that advantages similar to those in the uplink case can be obtained for the downlink case through the duality of a relay-aided interfering multiple-access channel (IMAC) and an interfering broadcast channel (IBC).
Furthermore, the proposed scheme for a full-duplex cellular network is shown to have advantages identical to those of a number of proposed half-duplex cellular cases.
\end{abstract}


\vspace{10pt}
\section{Introduction}

It is known that high degrees-of-freedom (DoF) can be achieved in a multi-user interference environment by means of interference alignment (IA). 
The optimal DoF of single-input single-output (SISO) time-varying interference channel has been derived through an asymptotic IA scheme \cite{cadambe08}. 
Research of IA in multiple-input multiple-output (MIMO) channels also has been done and the optimal DoF region for a multi-user MIMO interference channel has been revealed \cite{gou10}, showing that IA can also increase the DoF of MIMO communication systems. 

Unfortunately, conventional IA schemes cannot be directly applied to cellular networks because the channel state information (CSI) feedback to transmitter and time extension are required.
For example, a subspace IA scheme for a SISO uplink cellular network, finding that the optimal DoF of one can be achieved as the number of cellular users increases was suggested \cite{suh08}. 
However, the subspace IA scheme requires more time extensions as the number of users increases. 
Thus the transmission/reception delay of the information increases as the number of users in the network becomes large, which is not suitable for a cellular network.  
In general, user equipment (UE) is not practically able to handle/process beamforming by CSI feedback and communicating nodes require short transmission/reception delays. 
In order to reduce the CSI feedback and time extension requirements, relay-aided IA schemes have been proposed for wireless communication networks, where the relay handles the CSI instead of the transmitter, adding the potential to reduce the time extension requirement in the MIMO interference channel. 
An IA scheme with two time slots and without CSIT feedback for a MIMO interference channel was proposed \cite{relay13}, but the effective noise increases due to zero-forcing and global CSI knowledge at the receiver (CSIR) is required.
To overcome the abovementioned limitations on this scheme, an opposite directional interference alignment (ODIA) scheme with relay was introduced \cite{odia15}.
It was found that the ODIA can also achieve the optimal DoF in a symmetric MIMO interference channel.

Recently, researchers studied the achievable DoF for a cellular network where a base station (BS) is equipped with directional or omni-directional antennas when considering the cellular network topology \cite{cellular14} \cite{cellular14-2}.
In \cite{cellular15}, the optimal DoF and an achievable IA scheme were derived in a MIMO cellular network using earlier result in \cite{cadambe09}. 
However, this scheme also requires infinite time extensions to achieve promised DoF for a large network, which is not practical for a cellular network. Furthermore, a one-shot linear IA scheme was proposed based on subspace IA \cite{cellular16}, but finding and providing feedback with regard to the optimized variables to align the interferences may place a heavy load on the UEs. 

Several studies have focused on the potential role of relays in a cellular network with IA. 
An IA scheme for a two-cell downlink cellular network based on different half-duplex relaying schemes with the decode-and-forward (DF) protocol was proposed \cite{relay15}. 
However, it is known that the relay design becomes complex when the DF protocol is used in the relay \cite{relay08}. 
It was shown that in a quasi-static flat-fading environment, a MIMO cellular network can achieve higher DoF by using a full-duplex relay with finite time extension \cite{relay12}. 
In practice, relay echo due to the full-duplex operation at the relay can severely degrade the performance of the cellular network.

In this paper, in order to solve the problems of applying IA to MIMO cellular networks, we propose an interfering multiple-access channel (IMAC)-ODIA scheme for an uplink cellular network with a half-duplex amplify-and-forward (AF) relay, where only two time slots are needed regardless of the network size, to achieve interference alignment. 
It has the advantage of no time extension and the ability to be operated with various channel settings. 
Further, CSIT at the UEs are not required in the proposed scheme.
We evaluate the DoF performance of the proposed IMAC-ODIA scheme and show that for certain antenna configuration, the proposed scheme can achieve the optimal DoF for a MIMO cellular network.

Furthermore, we propose a relay-aided IA scheme for a downlink cellular network to resolve the global CSIR requirement in the conventional IA scheme which requires a heavy load on the UEs. 
Using the uplink-downlink duality from earlier work \cite{dual15}, we propose a modification of the IMAC-ODIA scheme for the downlink, which leads to interfering broadcast channel (IBC)-ODIA with features similar to those of the IMAC-ODIA scheme. 
Here, neither decorrelator construction nor global CSIR is required at the UEs.
We also extend our idea to full-duplex cellular network, referred to as full-duplex ODIA (FD-ODIA). 
It is important to handle interference to obtain full-duplexing gain and accordingly, we provide an IA beamformer design at each node and discuss its DoF performance.

This paper is organized as follows. In Section \ref{sec_preliminaries}, we present preliminary information concerning the properties of the tensor product and system model of the cellular network with a relay. In Section \ref{sec_IAcondition}, the IA condition and the IMAC-ODIA scheme for the uplink cellular network are proposed. 
An IBC-ODIA scheme for a downlink cellular network and its extension to a full-duplex cellular network are also proposed in Section \ref{sec_ibcodia} and Section \ref{sec_fdodia}, respectively. 
Section \ref{sec_discussion} includes the discussion, and finally the conclusion is given in Section \ref{sec_conclusion}.

\vspace{10pt}
\section{Preliminaries and System Model} \label{sec_preliminaries}

In this section, we introduce several properties of the tensor product of matrices \cite{tensor} and describe the cellular network model with an AF relay. 
First, we focus on an uplink symmetric cellular network modeled as a symmetric IMAC.

Throughout the paper, scalars are denoted in lowercase, vectors are written in boldface lowercase, and matrices are indicated by boldface capital letters. 
Several definitions pertaining to matrix \textbf{A} and vector \textbf{a} are given below.

- $a_{ij}$: $(i, j)$ component of matrix \textbf{A}

- $a_{i}$: $i$-th component of vector \textbf{a}

- $\textbf{null}^{\text{left}}_{\textbf{A}}$ and $\textbf{null}^{\text{right}}_{\textbf{A}}$: left and right null vectors of matrix \textbf{A}, respectively

- $\{\textbf{A}\}_{i}$: $i$-th column vector of matrix \textbf{A}

- $\{\textbf{A}\}_{i:j}$: submatrix of \textbf{A}, i.e., $[\{\textbf{A}\}_{i}, \{\textbf{A}\}_{i+1}, ..., \{\textbf{A}\}_{j}]$

- \text{rank}(\textbf{A}) or $r_\textbf{A}$: rank of matrix \textbf{A}

- ${(\cdot)}^T$, ${(\cdot)}^H$, ${(\cdot)}^{\text{left}}$, ${(\cdot)}^{\text{right}}$, and ${(\cdot)}^{\dagger}$: transpose, complex conjugated transpose, left inverse, 

\hspace{35pt} right inverse, and Moore-Penrose pseudo inverse, respectively 

- $\textbf{0}_{I\times J}$: $I\times J$ zero matrix and subscript can be omitted when the size of the zero matrix is 

\hspace{35pt} not important.

- $\textbf{I}_I$: $I\times I$ identity matrix

- $\textbf{I}_{I\times J}$: $I\times J$ rectangular identity matrix $=\setcounter{MaxMatrixCols}{1}
\begin{bmatrix}
\textbf{I}_J\\ \textbf{0}_{(I-J)\times J}
\end{bmatrix}$, $I\ge J$

- $\mathbb K$: real or complex field, $\mathbb R$ or $\mathbb C$

\subsection{Tensor Product and Kruskal Rank}
First, we briefly introduce the pre-defined tensor products of matrices, that is, the Kronecker and Khatri-Rao products. The Kronecker product of \textbf{A} and \textbf{B} is defined below.
\vspace{10pt}
\begin{definition} \label{def_kronecker}

$\textbf{A}\otimes \textbf{B}$ is the Kronecker product of \textbf{A} and \textbf{B} defined as

\setcounter{MaxMatrixCols}{3}
\begin{equation*}
\textbf{A}\otimes \textbf{B}=
\begin{bmatrix}
a_{11}\textbf{B}&a_{12}\textbf{B}&\cdots\\a_{21}\textbf{B}&a_{22}\textbf{B}&\cdots\\\vdots&\vdots&\ddots
\end{bmatrix}
.
\end{equation*}
\end{definition}
\vspace{10pt}

Let $\textbf{A}=[\textbf{A}_1\text{ }...\text{ }\textbf{A}_D]$ and $\textbf{B}=[\textbf{B}_1\text{ }...\text{ }\textbf{B}_D]$ be two partitioned matrices with an equal number of partitions. The Khatri-Rao product of \textbf{A} and \textbf{B} is then defined as the partition-wise Kronecker product as given below.
\vspace{10pt}
\begin{definition} \label{def_khatrirao}

$\textbf{A}\odot \textbf{B}$ is the Khatri-Rao product of partitioned matrices \textbf{A} and \textbf{B} defined as
\setcounter{MaxMatrixCols}{3}
\begin{equation*}
\textbf{A}\odot \textbf{B}=
\begin{bmatrix}
\textbf{A}_1\otimes \textbf{B}_1&...&\textbf{A}_D\otimes \textbf{B}_D
\end{bmatrix}
.
\end{equation*}
\end{definition}
\vspace{10pt}

Further, Kruskal rank and generalized Kruskal ranks are introduced as follows.
\vspace{10pt}
\begin{definition} \label{def_kruskal}

The Kruskal rank of matrix \textbf{A} denoted by $\text{rank}_k(\textbf{A})$ or $k_\textbf{A}$, is the maximal number \textit{r} such that any set of \textit{r} columns of \textbf{A} is linearly independent.

\end{definition}
\vspace{10pt}
\begin{definition} \label{def_gkruskal}

The generalized Kruskal rank of partitioned matrix $\textbf{A}=[\textbf{A}_1\text{ }...\text{ }\textbf{A}_D]$ denoted by $\text{rank}_{k'}(\textbf{A})$ or $k'_\textbf{A}$, is the maximal number \textit{r} such that any set of \textit{r} submatrices of \textbf{A} yields a set of linearly independent columns.

\end{definition}
\vspace{10pt}

We introduce a number of propositions for the generalized Kruskal rank from \cite{tensor}, which will be used in the next sections.
\vspace{10pt}
\begin{proposition}[Generalized Kruskal rank of an uniformly partitioned matrix \cite{tensor}]\label{prop_krank_uni}
{
Let $\textbf{A}=[\textbf{A}_1\text{ }...\text{ }\textbf{A}_D]$ be a matrix whose entries are drawn i.i.d. from a continuous
distribution and partitioned in $D$ submatrices with $\textbf{A}_r\in \mathbb K^{I\times L}$. 
Then, the generalized Kruskal rank of \textbf{A} is $\min\{{\lfloor {I\over L}\rfloor, D}\}$.
}
\end{proposition}
\vspace{10pt}
\begin{proposition} [Lemma 3.2 \cite{tensor}] \label{prop_krank}
{Consider the partitioned matrices $\textbf{A}=[\textbf{A}_1\text{ }...\text{ }\textbf{A}_D]$ with $\textbf{A}_r\in \mathbb K^{I \times L_r}$ and $\textbf{B}=[\textbf{B}_1\text{ }...\text{ }\textbf{B}_D]$ with $\textbf{B}_r\in \mathbb K^{J \times M_r}$, $1\le r\le D$. Then,

\begin{enumerate}

\item If $k'_\textbf{A}=0$ or $k'_\textbf{B}=0$, then $k'_\textbf{A$\odot$B}=0$.
\item If $k'_\textbf{A}\ge1$ and $k'_\textbf{B}\ge1$, then $k'_\textbf{A$\odot$B}\ge\min\{k'_\textbf{A}+k'_\textbf{B}-1, D\}$.

\end{enumerate}
}
\end{proposition}

\subsection{System Model: Cellular Network with a Single Relay}

In this paper, the information-theoretic quantity of interest is the DoF, where the DoF is defined as the number of successfully decodable data streams in the desired receiver that are transmitted by the corresponding user. 
The term ``successfully decodable'' means that the desired data streams are received in an interference-free space.

In a cellular network, the DoF per cell is of interest in general. 
It is defined as the interference-free dimension at the base station, regardless of whether it is the receiver or the transmitter. 
Once the DoF per cell is derived, it can be used to determine the number of users for data transmission per cell in a cellular network.

\begin{figure}
\centering
\includegraphics[scale=0.6]{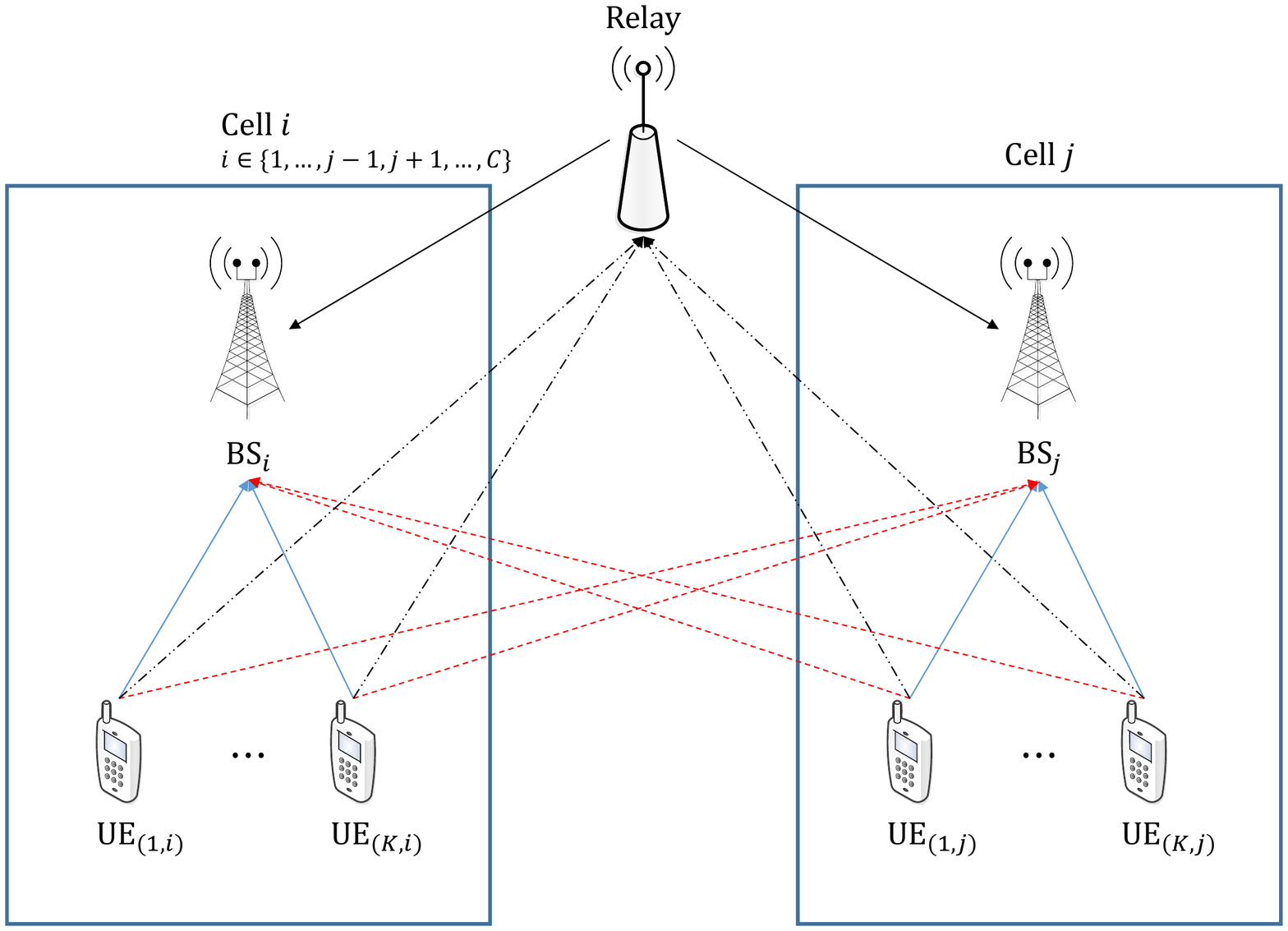}
\caption{System model: Fully connected cellular uplink network with a single relay.}
\label{fig_systemmodel_uplink}
\end{figure}

First, we focus on uplink cellular network, after which downlink cellular network will be considered. 
We consider a fully connected symmetric cellular network with $C$ cells and $K$ users in each cell, where `symmetric' means that each cell has an identical configuration.
Cell $j$ is served by a single base station denoted as $\text{BS}_j$, $j\in \{1, ..., C\}$. 
The user $(k,j)$ denoted as $\text{UE}_{(k,j)}$, $k\in \{1, ..., K\}$, is served by $\text{BS}_j$ as shown in Fig. \ref{fig_systemmodel_uplink}. 
It is assumed that each UE has $M$ transmit antennas, each BS has $N$ receive antennas, and the relay has $N_R$ antennas. 
The channel from $\text{UE}_{(k,i)}$ to $\text{BS}_j$ is denoted by the $N\times M$ matrix $\textbf{H}_{j,(k, i)}$, the channel from $\text{UE}_{(k,j)}$ to the relay is denoted by the $N_R\times M$ matrix $\textbf{H}^{UR}_{(k, j)}$, and the channel from the relay to $\text{BS}_j$  is denoted by the $N\times N_R$ matrix $\textbf{H}^{RB}_{j}$. 
In the proposed schemes, only the relay requires global CSI. 
In addition, each BS has its local CSI, but UEs do not require any CSI knowledge.

The channel is assumed to be Rayleigh fading and the entries of the channel matrices are independent identically distributed complex Gaussian random variables with zero-mean and unit-variance. 
Due to the average power constraint in cellular networks, the data vector $\textbf{x}_{(k, j)}$ transmitted from  $\text{UE}_{(k, j)}$ is normalized such that the average transmit power of each UE is limited by $P$, i.e., $\textbf{E}[||\textbf{x}_{(k, j)}||]\le P$, where $\textbf{E}[\cdot]$ and $||\cdot ||$ denote the expectation and norm functions, respectively. 

We consider a relay-aided half-duplexing communication scenario. 
Let $\textbf{y}_{j,1}$ and $\textbf{y}_{R}$ be the corresponding received signal vectors of $\text{BS}_j$ and the relay at the first time slot as

\begin{equation}\label{equ_rec_1st}
\textbf{y}_{j,1}=\sum_{k=1}^{K}{\textbf{H}_{j,(k, j)}\textbf{x}_{(k, j)}}+\sum_{i\ne j}^{C}{\sum_{k=1}^{K}{\textbf{H}_{j,(k, i)}\textbf{x}_{(k,i)}}}+\textbf{n}_{j,1}, \text{ } j\in \{1, ..., C\},
\vspace{10pt}
\end{equation}
\begin{equation}\label{equ_rec_rel}
\textbf{y}_{R}=\sum_{j=1}^{C}{\sum_{k=1}^{K}{\textbf{H}^{UR}_{(k, j)}}\textbf{x}_{(k, j)}}+\textbf{n}_{R}
,
\vspace{10pt}
\end{equation}
where $\textbf{n}_{j,1}\in \mathbb K^{N \times 1}$ and $\textbf{n}_{R}\in \mathbb K^{N_R \times 1}$ are zero-mean unit-variance circularly symmetric additive white complex Gaussian noise vectors at $\text{BS}_j$ and the relay, respectively. 
The received signal vector of $\text{BS}_j$ at the second time slot is given as

\begin{equation}\label{equ_rec_2nd_simple}
\textbf{y}_{j,2}=\textbf{H}^{RB}_{j}\textbf{T}\textbf{y}_R+\textbf{n}_{j,2}, \text{ } j\in \{1, ..., C\}
,
\vspace{10pt}
\end{equation}
where $\textbf{n}_{j,2}\in \mathbb K^{N \times 1}$ is also a zero-mean unit-variance circularly symmetric additive white complex Gaussian noise vector at $\text{BS}_j$ and $\textbf{T}\in \mathbb K^{N_R \times N_R}$ is the relay beamforming matrix. Relay beamformer is implemented based on the global CSI. 

Plugging (\ref{equ_rec_rel}) into (\ref{equ_rec_2nd_simple}) leads to

\begin{equation}\label{equ_rec_2nd}
\textbf{y}_{j,2}=\sum_{k=1}^{K}{\textbf{H}^{RB}_{j}\textbf{T}\textbf{H}^{UR}_{(k, j)}\textbf{x}_{(k,j)}}+\sum_{i\ne j}^{C}{\sum_{k=1}^{K}{\textbf{H}^{RB}_{j}\textbf{T}\textbf{H}^{UR}_{(k, i)}\textbf{x}_{(k,i)}}}+(\textbf{H}^{RB}_{j}\textbf{T}\textbf{n}_R+\textbf{n}_{j,2}), \text{ } j\in \{1, ..., C\}
.
\vspace{10pt}
\end{equation}

\section{IMAC-ODIA for Uplink Cellular Network} \label{sec_IAcondition}

First, we describe the key idea of the proposed IMAC-ODIA scheme for an uplink cellular network, after which we simplify (\ref{equ_rec_1st}) and (\ref{equ_rec_2nd}) using augmented channel matrices, which will be used for the proposed scheme. 
The IA conditions to be satisfied for the proposed scheme will also be described.
Finally, we present the details of relay beamformer which meet the IA conditions and analyze the achievable DoF of an uplink cellular network with the proposed IMAC-ODIA scheme.

\begin{figure}
\centering
\includegraphics[scale=0.8]{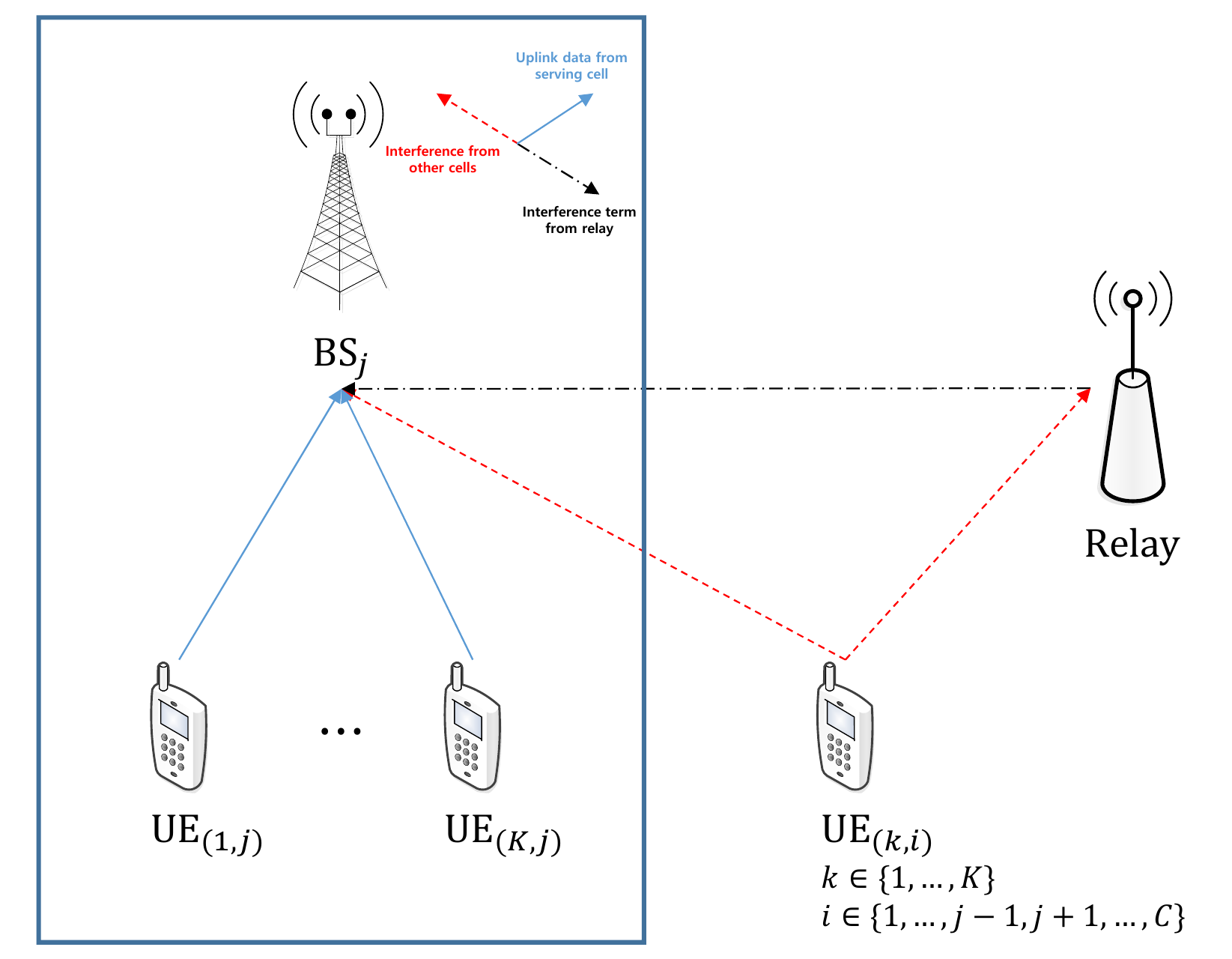}
\caption{Interference alignment by a relay in an uplink cellular network.}
\label{fig_IA_description}
\end{figure}

\subsection{IMAC-ODIA for a Symmetric Uplink Cellular Network}

The conventional IA is to align interferences in a subspace which is linearly independent of the desired signal space by designing a beamformer at the transmitter. 
Then, using a decorrelator on the receiver side, the aligned interferences can be perfectly cancelled. 
However, designing a beamformer at the transmitter requires CSIT and increases the computational complexity. 
In a cellular system, UEs are not likely to be capable of beamforming.

Thus, we propose an IMAC-ODIA scheme for an uplink cellular network, where the interference vector received from a direct channel in the first time slot and the interference vector received from the relay in the second time slot sum to zero. 
This is achieved by designing the beamforming matrix only at the relay as in Fig. \ref{fig_IA_description}. 
This implies that neither CSIT nor complex computing is required at the UEs.

Using augmentation of the matrices and vectors in (\ref{equ_rec_1st}) and (\ref{equ_rec_rel}), the abovementioned IA description can be simplified for the proposed IMAC-ODIA scheme. 
The augmented channel matrices for $\text{BS}_j$ are denoted as follows. 

- $\bar{\textbf{H}}_j$: interfering channel matrix from UEs of the other cells to $\text{BS}_j$

- $\bar{\textbf{H}}^{UR}_j$: interfering channel matrix from UEs of the other cells to the relay

- $\widehat{\textbf{H}}_j$: desired channel matrix from the UEs served by $\text{BS}_j$ to $\text{BS}_j$

- $\widehat{\textbf{H}}^{UR}_j$: desired channel matrix from the UEs served by $\text{BS}_j$ to the relay \\
These are defined as follows.
\begin{align}
\setcounter{MaxMatrixCols}{6}
\label{H_bar}
\bar{\textbf{H}}_j&=
\begin{bmatrix}
\textbf{H}_{j,(1,1)}&...&\textbf{H}_{j,(K,j-1)}&\textbf{H}_{j,(1,j+1)}&...&\textbf{H}_{j,(K,C)}
\end{bmatrix}\in \mathbb K^{N \times (C-1)KM}, \\
\nonumber \\
\setcounter{MaxMatrixCols}{6}
\bar{\textbf{H}}^{UR}_j&=
\begin{bmatrix}
\textbf{H}^{UR}_{(1, 1)}&...&\textbf{H}^{UR}_{(K, j-1)}&\textbf{H}^{UR}_{(1, j+1)}&...&\textbf{H}^{UR}_{(K, C)}
\end{bmatrix}\in \mathbb K^{N_R \times (C-1)KM}, \\
\nonumber \\
\setcounter{MaxMatrixCols}{3}
\widehat{\textbf{H}}_j&=
\begin{bmatrix}
\textbf{H}_{j,(1,j)}&...&\textbf{H}_{j,(K,j)}
\end{bmatrix}\in \mathbb K^{N \times KM}, \\
\nonumber \\
\setcounter{MaxMatrixCols}{3}
\label{H_UR_hat}
\widehat{\textbf{H}}^{UR}_j&=
\begin{bmatrix}
\textbf{H}^{UR}_{(1, j)}&...&\textbf{H}^{UR}_{(K, j)}
\end{bmatrix}\in \mathbb K^{N_R \times KM}.
\end{align}
Here, let $\bar{\textbf{x}}_j$ and $\hat{\textbf{x}}_j$ be the interfering data stream and the desired data stream received at $\text{BS}_j$, respectively.
They should also be converted into an augmented form as

\setcounter{MaxMatrixCols}{1}
\begin{equation*}
\bar{\textbf{x}}_j=
\begin{bmatrix}
\textbf{x}_{(1,1)}\\ \vdots\\ \textbf{x}_{(K,j-1)}\\ \textbf{x}_{(1,j+1)}\\ \vdots\\ \textbf{x}_{(K,C)}
\end{bmatrix}\in \mathbb K^{(C-1)KM \times 1}
, \hspace{70pt}
\hat{\textbf{x}}_j=
\begin{bmatrix}
\textbf{x}_{(1,j)}\\ \vdots\\ \textbf{x}_{(K,j)}
\end{bmatrix}\in \mathbb K^{KM \times 1}.
\vspace{10pt}
\end{equation*}
From the augmented matrices and vectors, (\ref{equ_rec_1st}) and (\ref{equ_rec_2nd}) can be rewritten as
\begin{align}
\label{equ_bs_j_1}
\textbf{y}_{j,1}&=\widehat{\textbf{H}}_j\hat{\textbf{x}}_j+\bar{\textbf{H}}_j\bar{\textbf{x}}_j+\textbf{n}_{j,1},\text{ } j\in \{1, ..., C\},
\\
\nonumber \\
\label{equ_bs_j_2}
\textbf{y}_{j,2}&=\textbf{H}^{RB}_{j}\textbf{T}(\widehat{\textbf{H}}^{UR}_j\hat{\textbf{x}}_j+\bar{\textbf{H}}^{UR}_j\bar{\textbf{x}}_j)+(\textbf{H}^{RB}_{j}\textbf{T}\textbf{n}_R+\textbf{n}_{j,2}),\text{ } j\in \{1, ..., C\}
.
\end{align}

Ignoring the noise term and adding (\ref{equ_bs_j_1}) and (\ref{equ_bs_j_2}), the total received signal vector at $\text{BS}_j$ becomes
\begin{equation}\label{equ_IA_1&2}
\textbf{y}_{j,1}+\textbf{y}_{j,2}=
(\widehat{\textbf{H}}_j+\textbf{H}^{RB}_{j}\textbf{T}\widehat{\textbf{H}}^{UR}_j)\hat{\textbf{x}}_j+(\bar{\textbf{H}}_j+\textbf{H}^{RB}_{j}\textbf{T}\bar{\textbf{H}}^{UR}_j)\bar{\textbf{x}}_j,\text{ } j\in \{1, ..., C\}
.
\vspace{10pt}
\end{equation}
In the proposed IMAC-ODIA scheme, the second term on the right-hand side in (\ref{equ_IA_1&2}) is the total received interference signal, which should be cancelled.
Thus, the IMAC-ODIA condition can be described as
\begin{equation}\label{equ_IA}
\bar{\textbf{H}}_j+\textbf{H}^{RB}_{j}\textbf{T}\bar{\textbf{H}}^{UR}_j=0,\text{ } j\in \{1, ..., C\}
.
\vspace{10pt}
\end{equation}
This indicates the necessity of only the local CSIR at $\text{BS}_j$ rather than the global CSIT and CSIR, though global CSI is required at the relay. 
At this point, our goal is to find the relay beamformer $\textbf{T}$ which meets the requirements in (\ref{equ_IA}). 
The design of the relay beamformer $\textbf{T}$ and its existence condition will be given in the next subsection.

\subsection{Existence of a Relay Beamformer and Its Design}\label{sub_ExtRel}

First, we introduce the important properties of the Kronecker product as in the following propositions for the main theorem in this subsection.
\vspace{10pt}

\begin{proposition} [Representation of the matrix by the Kronecker product] \label{prop_kronec}
{
Consider matrices \textbf{A}, \textbf{B}, \textbf{C}, and \textbf{X} and the matrix equation
\begin{equation} \label{equ_kron_example}
\textbf{A}\textbf{X}\textbf{B}=\textbf{C}
.
\end{equation}
Here, (\ref{equ_kron_example}) can be transformed by vectorization into
\begin{equation*} \label{equ_kron_vec}
\text{vec}(\textbf{A}\textbf{X}\textbf{B})=(\textbf{B}^T\otimes \textbf{A})\text{vec}(\textbf{X}) =\text{vec}(\textbf{C})
,
\end{equation*}
where \text{vec}(\textbf{X}) denotes the vectorization of matrix \textbf{X} by stacking the columns of \textbf{X} into a single column vector. 
}
\vspace{10pt}
\end{proposition}
\begin{proposition} [Transpose of the Kronecker product] \label{prop_kronec_trans}
{
For matrices \textbf{A} and \textbf{B}, the following property holds, that is,
\begin{equation*} \label{equ_kron_trans}
(\textbf{A}\otimes \textbf{B})^T=\textbf{A}^T\otimes \textbf{B}^T
.
\end{equation*}
}
\end{proposition}
Then, the following theorem provides the lower bound of the required number of antennas at the relay to implement IMAC-ODIA for a symmetric uplink cellular network.

\vspace{10pt}
\begin{theorem}[Required number of antennas at the relay for beamforming]\label{the_exist_beamformer}
The relay beamformer $\textbf{T}$ satisfying condition (\ref{equ_IA}) exists if $N_R\ge \max\{{(C-1)KM,CN}\}$
.
\end{theorem}
\vspace{10pt}

\begin{proof}
From \textit{Proposition \ref{prop_kronec}}, (\ref{equ_IA}) can be transformed into

\begin{equation} \label{equ_equi_IA}
\{(\bar{\textbf{H}}^{UR}_j)^T\otimes \textbf{H}^{RB}_{j}\} \text{vec}(\textbf{T})=\text{vec}(-\bar{\textbf{H}}_j),\text{ } j\in \{1, ..., C\}
.
\vspace{10pt}
\end{equation}
The $C$ equations in (\ref{equ_equi_IA}) can be rewritten as a single equation 

\setcounter{MaxMatrixCols}{1}
\begin{equation} \label{equ_IA_mat}
\begin{bmatrix}
(\bar{\textbf{H}}^{UR}_1)^T\otimes \textbf{H}^{RB}_{1}\\ (\bar{\textbf{H}}^{UR}_2)^T\otimes \textbf{H}^{RB}_{2}\\ \vdots\\ (\bar{\textbf{H}}^{UR}_C)^T\otimes \textbf{H}^{RB}_{C}
\end{bmatrix}
\text{vec}(\textbf{T})=
\begin{bmatrix}
\text{vec}(-\bar{\textbf{H}}_1)\\ \text{vec}(-\bar{\textbf{H}}_2)\\ \vdots\\ \text{vec}(-\bar{\textbf{H}}_C)
\end{bmatrix}
.
\vspace{10pt}
\end{equation}

Let $\textbf{H}=\begin{bmatrix}
(\bar{\textbf{H}}^{UR}_1)^T\otimes \textbf{H}^{RB}_{1}\\ (\bar{\textbf{H}}^{UR}_2)^T\otimes \textbf{H}^{RB}_{2}\\ \vdots\\ (\bar{\textbf{H}}^{UR}_C)^T\otimes \textbf{H}^{RB}_{C}
\end{bmatrix}$ and $\textbf{h}=\begin{bmatrix}
\text{vec}(-\bar{\textbf{H}}_1)\\ \text{vec}(-\bar{\textbf{H}}_2)\\ \vdots\\ \text{vec}(-\bar{\textbf{H}}_C)
\end{bmatrix}
$. 
Then, (\ref{equ_IA_mat}) can be represented as 

\vspace{10pt}
\begin{equation} \label{equ_HTh}
 \textbf{H}\times \text{vec}(\textbf{T})=\textbf{h}
.
\vspace{10pt}
\end{equation}

Our goal is to find \text{vec}(\textbf{T}) satisfying (\ref{equ_HTh}). First, we have to show its existence.
Solutions of (\ref{equ_HTh}) exist if the following condition is satisfied as 

\setcounter{MaxMatrixCols}{2}
\begin{equation}\label{equ_rank}
\text{rank}
(\begin{bmatrix}
\textbf{H}&\textbf{h}
\end{bmatrix})
=
\text{rank}(\textbf{H})
.
\vspace{10pt}
\end{equation}
Since all entries of \textbf{h} are drawn i.i.d. from a continuous distribution, (\ref{equ_rank}) can be satisfied if and only if \textbf{H} has full row rank, which means that $\textbf{H}^T$ has full column rank.
From \textit{Proposition \ref{prop_kronec_trans}}, $\textbf{H}^T$ can be rewritten as

\setcounter{MaxMatrixCols}{4}
\begin{equation}\label{khatri_example}
\textbf{H}^T=
\begin{bmatrix}
\bar{\textbf{H}}^{UR}_1\otimes (\textbf{H}^{RB}_{1})^T &\bar{\textbf{H}}^{UR}_2\otimes (\textbf{H}^{RB}_{2})^T &... &\bar{\textbf{H}}^{UR}_C\otimes (\textbf{H}^{RB}_{C})^T
\end{bmatrix}
.
\vspace{10pt}
\end{equation}
In fact, $\textbf{H}^T$ in (\ref{khatri_example}) is the Khatri-Rao product of the two matrices $\bar{\textbf{H}}^{UR}$ and $\textbf{H}^{RB}$, where

\setcounter{MaxMatrixCols}{4}
\begin{equation*}\label{khatri_example2}
\bar{\textbf{H}}^{UR}=
\begin{bmatrix}
\bar{\textbf{H}}^{UR}_1& \bar{\textbf{H}}^{UR}_2& ...&\bar{\textbf{H}}^{UR}_C
\end{bmatrix}
, \hspace{10pt}
\textbf{H}^{RB}=
\begin{bmatrix}
(\textbf{H}^{RB}_{1})^T& (\textbf{H}^{RB}_{2})^T& ...& (\textbf{H}^{RB}_{C})^T
\end{bmatrix}
.
\vspace{10pt}
\end{equation*}
Note that $\bar{\textbf{H}}^{UR}\in \mathbb K^{N_R \times C(C-1)KM}$ and $\textbf{H}^{RB}\in \mathbb K^{N_R \times CN}$.

\begin{figure}
\centering
\includegraphics[scale=0.7]{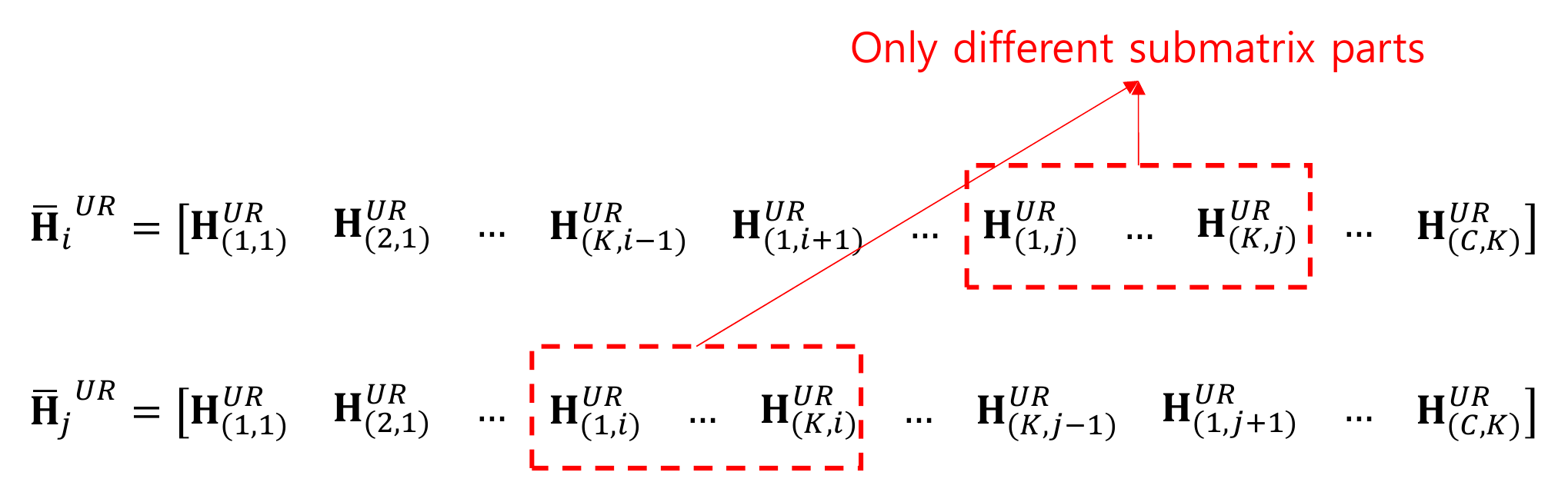}
\caption{Matrix structure of $\bar{\textbf{H}}^{UR}_{i}$ and $\bar{\textbf{H}}^{UR}_{j}$ for $i\neq j$.}
\label{fig_matrix_URbar}
\end{figure}

From Fig. \ref{fig_matrix_URbar}, there are precisely $(C-2)K$ identical $N_R\times M$ submatrices for any pair of matrices in $\{ \bar{\textbf{H}}^{UR}_1, \bar{\textbf{H}}^{UR}_2, ..., \bar{\textbf{H}}^{UR}_C\}$ with probability one. 
This implies that the generalized Kruskal rank of $\bar{\textbf{H}}^{UR}$ denoted by $k'_{\bar{\textbf{H}}^{UR}}$ has a maximum value of one.  
From \textit{Proposition \ref{prop_krank_uni}}, $k'_{\bar{\textbf{H}}^{UR}}$ can have two different values 
{
\begin{equation*}
k'_{\bar{\textbf{H}}^{UR}}=\min{\{ \lfloor {N_R\over (C-1)KM}\rfloor, C, 1\}}=
\begin{cases}
1, \;{\rm if}\; N_R\ge (C-1)KM \\
0, \;{\rm if}\; N_R< (C-1)KM
.
\end{cases}
\vspace{10pt}
\end{equation*}
}
Furthermore, because $\textbf{H}^{RB}$ has submatrices, whose entries are drawn i.i.d. from a continuous distribution, $k'_{\textbf{H}^{RB}}$ is also given as

{
\begin{equation*}
k'_{\textbf{H}^{RB}}=\min{\{ \lfloor {N_R\over N}\rfloor, C\}}=
\begin{cases}
C,\hspace{17pt} \;{\rm if}\; N_R\ge CN \\
\lfloor {N_R\over N}\rfloor, \;{\rm if}\; N_R< CN
.
\end{cases}
\vspace{10pt}
\end{equation*}
}

From \textit{Proposition \ref{prop_krank}}, we can find the lower bound on the generalized Kruskal rank of $\textbf{H}^T=\bar{\textbf{H}}^{UR}\odot \textbf{H}^{RB}$. As long as $k'_{\bar{\textbf{H}}^{UR}}$ and $k'_{\textbf{H}^{RB}}$ are not zeros, lower bound of $k'_{\textbf{H}^T}=k'_{\bar{\textbf{H}}^{UR}\odot \textbf{H}^{RB}}$ is given as
{
\begin{equation*}
k'_{\textbf{H}^T}=k'_{\bar{\textbf{H}}^{UR}\odot \textbf{H}^{RB}}\ge \min{\{1+C-1, C\}}=C, \;{\rm if}\; N_R\ge \max{\{ (C-1)KM, CN\}}
.
\vspace{10pt}
\end{equation*}
}Considering that $\textbf{H}^T$ is partitioned into $C$ submatrices, clearly, $k'_{\textbf{H}^T}\le C$. 
Thus, the value of $k'_{\textbf{H}^T}$ is obtained as $C$, which means that $\textbf{H}^T$ has full column rank and (\ref{equ_HTh}) has the solution for \text{vec}(\textbf{T}). 
Thus, we prove the theorem. 
\end{proof}

Note that \textbf{H} has its right inverse $\textbf{H}^{\text{right}}=\textbf{H}^T\{\textbf{H}\textbf{H}^T\}^{-1}$.
Then, the relay beamformer \textbf{T} is given as
\setcounter{MaxMatrixCols}{1}
\begin{equation} \label{equ_relay_beam}
\textbf{T}=\text{vec}_{N_R}^{-1}(\textbf{null}^{\text{right}}_{\textbf{H}}+\textbf{H}^{\text{right}}\textbf{h})
,
\vspace{10pt}
\end{equation}
where $\text{vec}_{N_R}^{-1}(\cdot)$ is defined as a devectorization function which transforms the vector into a matrix with a column size of $N_R$.
For simplification, $\textbf{null}^{\text{right}}_{\textbf{H}}$ can be set as a zero vector.

\subsection{Achievable DoF of a Symmetric Uplink Cellular Network with IMAC-ODIA}\label{sub_ExtRel}

First, we present the result of this subsection, which describes the achievable DoF per cell for an uplink cellular network with IMAC-ODIA.
\vspace{10pt}
\begin{theorem}[Achievable DoF per cell for an uplink cellular network with IMAC-ODIA]\label{the_dof_cell}
The achievable DoF per cell of a $(C,K,M,N)$ uplink cellular network with IMAC-ODIA is given as 
\begin{equation*}\label{equ_dof}
\text{DoF}_{\text{cell}}=
{
\min\{N, KM\}\over 2
}
.
\vspace{10pt}
\end{equation*}
\end{theorem}

\begin{proof}
With the relay beamformer $\textbf{T}$ in (\ref{equ_relay_beam}), the inter-cell interferences can be aligned at the relay and cancelled at each receiver node. 
Thus, the DoF of $\text{BS}_j$ for two transmission time slots can be simply calculated as the rank of the effective channel of the desired data vector $\hat{\textbf{x}}_j$ and divided by two. 
From (\ref{equ_IA_1&2}) and (\ref{equ_IA}), the effective channel $\widehat{\textbf{H}}_{\text{eff},j}$ of $\text{BS}_j$ for the desired signal is defined as 
$\widehat{\textbf{H}}_{\text{eff},j}\hat{\textbf{x}}_j=\textbf{y}_{j,1}+\textbf{y}_{j,2}$ and its DoF are given as
\begin{equation}\label{equ_chan_eff}
\widehat{\textbf{H}}_{\text{eff}, j}=\widehat{\textbf{H}}_j+\textbf{H}^{RB}_{j}\textbf{T}\widehat{\textbf{H}}^{UR}_j,
\vspace{10pt}
\end{equation}
\begin{equation*}
\text{DoF}_{j}=
{r_{\widehat{\textbf{H}}_{\text{eff},j}}\over 2}
.
\vspace{10pt}
\end{equation*}

Since the left inverse of $\bar{\textbf{H}}^{UR}_j$ exists as $\{(\bar{\textbf{H}}^{UR}_j)^{T}\bar{\textbf{H}}^{UR}_j\}^{-1}(\bar{\textbf{H}}^{UR}_j)^{T}$, (\ref{equ_IA}) can be transformed into $\textbf{H}^{RB}_{j}\textbf{T}=-\bar{\textbf{H}}_j(\bar{\textbf{H}}^{UR}_j)^{\text{left}}+\textbf{null}^{\text{left}}_{\bar{\textbf{H}}^{UR}_j}$.
Therefore, (\ref{equ_chan_eff}) can be rewritten as

\begin{equation}\label{equ_chan_eff2}
\widehat{\textbf{H}}_{\text{eff},j}=\widehat{\textbf{H}}_j-\bar{\textbf{H}}_j(\bar{\textbf{H}}^{UR}_j)^{\text{left}}\widehat{\textbf{H}}^{UR}_j+\textbf{null}^{\text{left}}_{\bar{\textbf{H}}^{UR}_j}\widehat{\textbf{H}}^{UR}_j \in \mathbb K^{N\times KM}
.
\vspace{10pt}
\end{equation}
Considering that the entries of each submatrix in the augmented matrices $\widehat{\textbf{H}}_j, \bar{\textbf{H}}_j, (\bar{\textbf{H}}^{UR}_j)^{\text{left}}, \text{and }  \widehat{\textbf{H}}^{UR}_j$ are drawn i.i.d. from a continuous distribution, the rank of (\ref{equ_chan_eff2}) can be derived with probability one as
\begin{equation*}
r_{\widehat{\textbf{H}}_{\text{eff},j}}=\min\{N, KM\}
.
\vspace{10pt}
\end{equation*}
Finally, the achievable DoF per cell for a symmetric uplink cellular network with IMAC-ODIA is given as
\begin{equation*}
\text{DoF}_{\text{cell}}=
{
\min\{N, KM\}\over 2
}
.
\end{equation*}
\end{proof}

\subsection{IMAC-ODIA for an Asymmetric Uplink Cellular Network}

The IMAC-ODIA for an asymmetric uplink cellular network is described in this subsection, where `asymmetric' means that each cell has a different configuration. 
Consider an uplink cellular network with $C$ cells and $K_j$ users served by a base station $\text{BS}_j$ with $N_j$ antennas for $j\in \{1, ..., C\}$. 
User $(k,j)$ is denoted by $\text{UE}_{(k,j)},\text{ } k\in \{1, ..., K_j\} \text{, which is served by } \text{BS}_j$. $\text{UE}_{(k,j)}$ is assumed to have $M_{(k,j)}$ antennas. 
Similarly, a relay is assumed to have $N_R$ antennas, and all channel parameters and matrices are defined in the same manner as the symmetric case with augmented matrices (\ref{H_bar}) through (\ref{H_UR_hat}). 
The number of required relay antennas and the DoF of each cell are presented in the form of the following theorem.
\vspace{10pt}
\begin{theorem}[IMAC-ODIA for an asymmetric uplink cellular network]
The relay beamformer of IMAC-ODIA for an asymmetric uplink cellular network exists if the following inequality is satisfied
\begin{equation*}\label{equ_asy}
N_R\ge \max\{{(\max_{j\in \{1, ..., C\}}\sum_{i\ne j}\sum_{k=1}^{K_i}{M_{(k,i)}}),(\sum_{j=1}^{C}N_j})\}
\vspace{10pt}
\end{equation*}
and the achievable DoF for $\text{BS}_j$ is given as
\begin{equation*}\label{equ_asy_dof}
\text{DoF}_j=
{
{1\over 2}\min\{N_j, \sum_{k=1}^{K_j}{M_{(k,j)}}\}
}
.
\vspace{10pt}
\end{equation*}
 \end{theorem}
\vspace{10pt}
The proof is identical to that of the symmetric case, and thus it is omitted here. 

It is clear that many features of IMAC-ODIA make it possible for a cellular network to adopt interference alignment.
Further, we extend the proposed scheme to a downlink cellular network with the help of uplink-downlink duality for relay-aided cellular networks.
This is discussed in the next section. 

\section{IBC-ODIA for Donwlink Cellular Network} \label{sec_ibcodia}

In this section, we extend the idea to a downlink cellular network, which can be modeled as an IBC. Here, we only consider a symmetric downlink cellular network because the IBC-ODIA scheme for an asymmetric case can be described in the same manner as the IMAC-ODIA scheme for an asymmetric cellular network. 
Channel reciprocity refers to the relationship between uplink and downlink channels, that is, each downlink channel matrix for a pair of nodes is a complex conjugate transpose of the corresponding uplink channel matrix for the same node pair.

\begin{figure}
\centering
\subfigure[Block diagram of uplink transmission with IMAC-ODIA]{\includegraphics[scale=0.5]{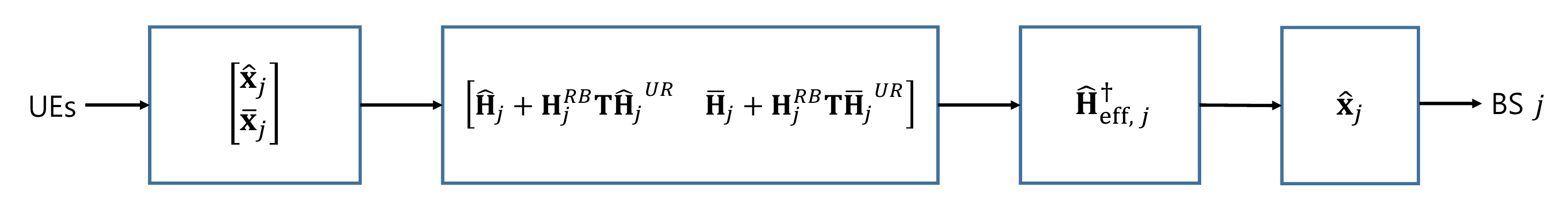}\label{fig_block_uplink}}
\subfigure[Complex conjugate transpose of the uplink transmission]{\includegraphics[scale=0.5]{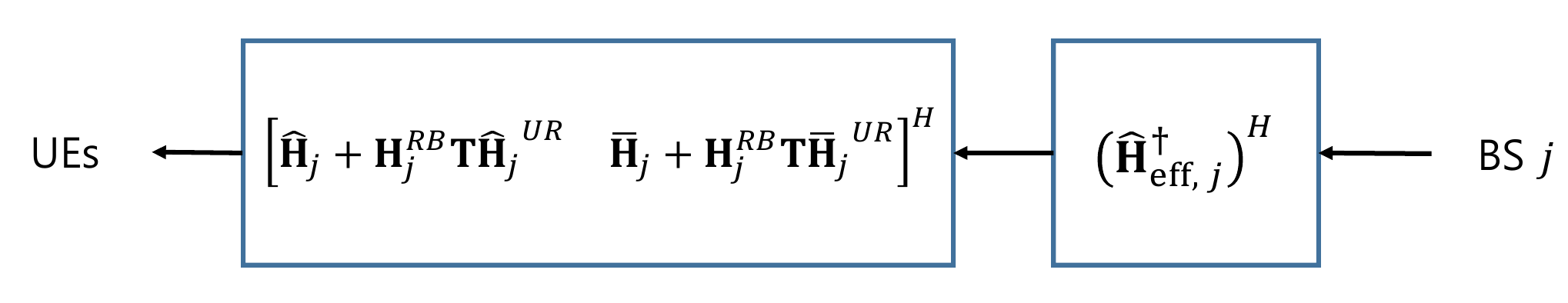}\label{fig_block_downlink}}
\caption{Block diagrams of the proposed IMAC-ODIA scheme.}
\label{fig_block}
\end{figure}

\subsection{System Model and Inter-Cell Interference Alignment}

We initially apply the result pertaining to duality from earlier work \cite{dual15} to the uplink cellular network with IMAC-ODIA, which gives us the motivation for IBC-ODIA for a downlink cellular network. For the uplink case, the desired data stream $\hat{\textbf{x}}_j$ for $\text{BS}_j$ can be obtained from (\ref{equ_IA_1&2}) and (\ref{equ_chan_eff}) as
\begin{equation}\label{equ_deco}
\hat{\textbf{x}}_j=\widehat{\textbf{H}}_{\text{eff},j}^{\dagger}(\textbf{y}_{j,1}+\textbf{y}_{j,2})
,
\vspace{10pt}
\end{equation}
where $\widehat{\textbf{H}}_{\text{eff},j}^{\dagger}$ serves as a decorrelator which handles intra-cell interference. 
In Fig. \ref{fig_block_uplink}, the uplink transmission process with IMAC-ODIA is described. 
If we assume channel reciprocity by a complex conjugate transpose operation, Fig. \ref{fig_block_uplink} is transformed into Fig. \ref{fig_block_downlink} without considering the data stream. 
This provides the intuition for IBC-ODIA, where transmitter (BS) beamforming may be required for the IBC-ODIA scheme for a downlink cellular network, but no decorrelator is needed at the receiver (UE). 
Therefore, IBC-ODIA does not require any complex operation at the UE compared to BS, which supports the feasibility of the proposed IBC-ODIA scheme for a downlink cellular network.

First, we describe the symmetric downlink cellular network, that is, $C$ cells and $K$ users in each cell. Cell $j$ is served by a single base station $\text{BS}_j$ for $j\in \{1, ..., C\}$. User $(k,j)$ is denoted by $\text{UE}_{(k,j)}$, $k\in \{1, ..., K\} \text{, which is served by } \text{BS}_j$ as in Fig. \ref{fig_systemmodel_downlink}. 
Each UE, BS, and relay is assumed to have $M$, $N$, and $N_R$ antennas, respectively. 
The channel from $\text{BS}_j$ to $\text{UE}_{(k,i)}$ is denoted by the $M\times N$ matrix $\textbf{H}_{(k, i),j}$, the channel from $\text{BS}_{j}$ to the relay is denoted by the $N_R\times N$ matrix $\textbf{H}^{BR}_{j}$, and the channel from relay to $\text{UE}_{(k,j)}$  is denoted by the $M\times N_R$ matrix $\textbf{H}^{RU}_{(k,j)}$. $\textbf{x}_{j}$ denotes the transmitted data vector from $\text{BS}_j$. 
Similarly to the IMAC-ODIA scheme, the augmented matrices for the IBC-ODIA scheme are defined as

\begin{figure}
\centering
\includegraphics[scale=0.5]{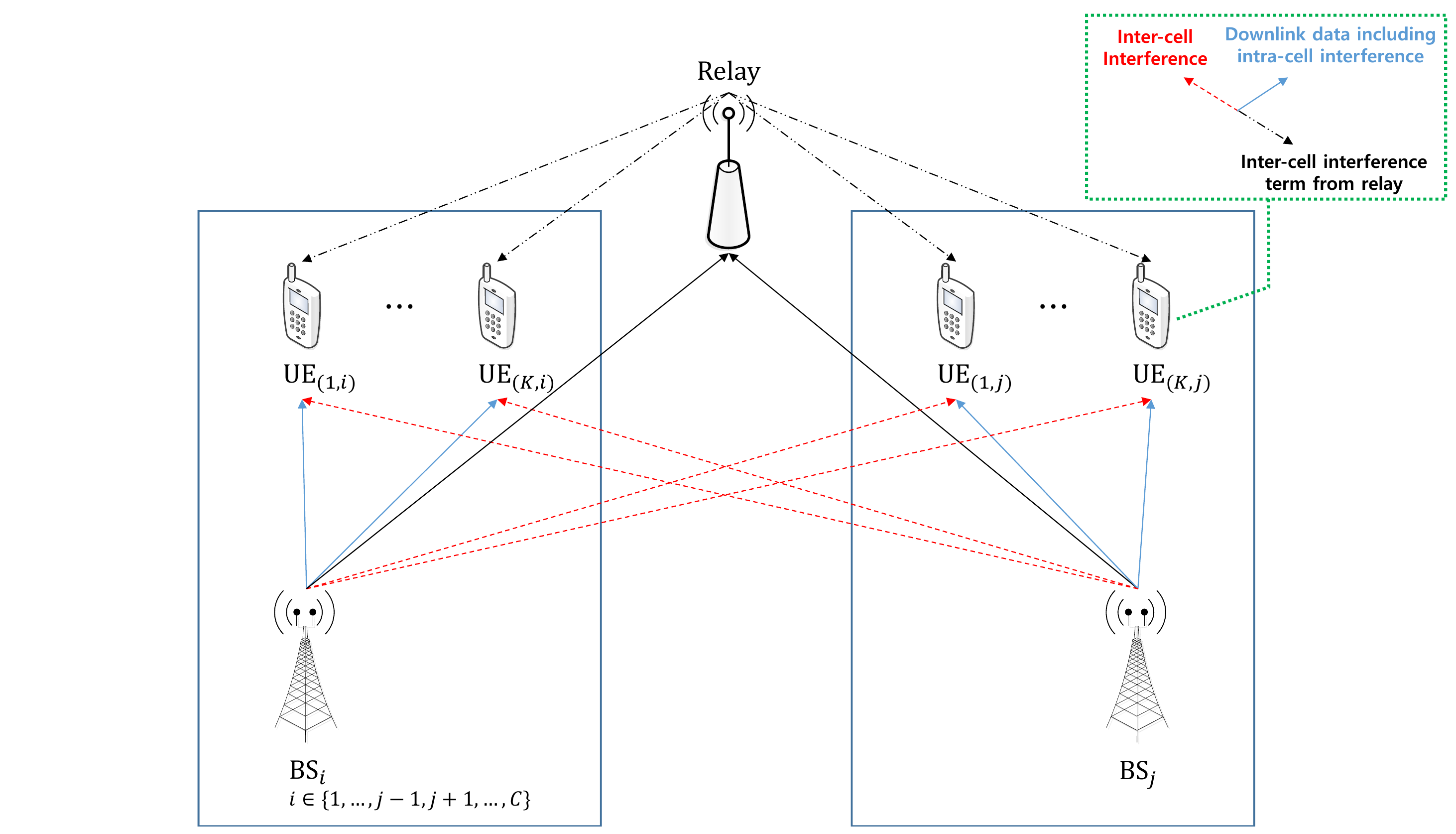}
\caption{Downlink cellular network with IBC-ODIA.}
\label{fig_systemmodel_downlink}
\end{figure}
\begin{align}
\setcounter{MaxMatrixCols}{6}
\bar{\textbf{H}}_{(k,j)}&=
\begin{bmatrix}
\textbf{H}_{(k,j),1}&...&\textbf{H}_{(k,j),j-1}&\textbf{H}_{(k,j),j+1}&...&\textbf{H}_{(k,j),C}
\end{bmatrix}\in \mathbb K^{M \times (C-1)N}, \nonumber \\
\nonumber \\
\setcounter{MaxMatrixCols}{6}
\bar{\textbf{H}}^{BR}_{j}&=
\begin{bmatrix}
\textbf{H}^{BR}_{1}&...&\textbf{H}^{BR}_{j-1}&\textbf{H}^{BR}_{j+1}&...&\textbf{H}^{BR}_{C}
\end{bmatrix}\in \mathbb K^{N_R \times (C-1)N} \nonumber \\ \nonumber
\end{align}
and the inter-cell interference vector for UEs at cell $j$ is defined as

\begin{equation}
\setcounter{MaxMatrixCols}{3}
\bar{\textbf{x}}_j=
\begin{bmatrix}
\textbf{x}_{1}\\ \vdots\\ \textbf{x}_{j-1}\\ \textbf{x}_{j+1}\\ \vdots\\ \textbf{x}_{C}
\end{bmatrix}\in \mathbb K^{(C-1)N \times 1}. \nonumber
\end{equation}

Similarly to the IMAC-ODIA case, by ignoring the noise term, the total received vector at $\text{UE}_{(k,j)}$ for two time slots is given as

\begin{equation*}\label{equ_IA_IBC}
\textbf{y}_{(k,j),1}+\textbf{y}_{(k,j),2}=
(\textbf{H}_{(k,j),j}+\textbf{H}^{RU}_{(k,j)}\textbf{T}_{\text{DL}}\textbf{H}^{BR}_j)\textbf{x}_j+(\bar{\textbf{H}}_{(k,j)}+\textbf{H}^{RU}_{(k,j)}\textbf{T}_{\text{DL}}\bar{\textbf{H}}^{BR}_{j})\bar{\textbf{x}}_j
,
\vspace{10pt}
\end{equation*}
where $\textbf{T}_{\text{DL}}$ represents the relay beamforming matrix for downlink transmission.
Note that $\textbf{T}_{\text{DL}}$ only aligns the inter-cell interferences. 
The remaining intra-cell interferences will be aligned by beamforming at the BS, which will be described later. The BS beamforming for intra-cell IA originates from the intuition for uplink-downlink duality as discussed earlier.

For $k\in \{1, ..., K\}$ and $i\in \{1, ..., C\}$, the interference alignment condition should be satisfied as
\begin{equation}\label{equ_IA_cond_DL}
\bar{\textbf{H}}_{(k,j)}+\textbf{H}^{RU}_{(k,j)}\textbf{T}_{\text{DL}}\bar{\textbf{H}}^{BR}_{j}=0.
\vspace{10pt}
\end{equation}
Further, $\textbf{T}_{\text{DL}}$ satisfying the above condition is given as 

\setcounter{MaxMatrixCols}{1}
\begin{equation*} \label{equ_relay_beam_DL}
\textbf{T}_{\text{DL}}=\text{vec}_{N_R}^{-1}\left(
\textbf{null}^{\text{right}}_{\textbf{H}_{\text{DL}}}+(\textbf{H}_{\text{DL}})^{\text{right}}\textbf{h}_{\text{DL}}
\right),
\vspace{10pt}
\end{equation*}
where $\textbf{H}_{\text{DL}}=\begin{bmatrix}
(\bar{\textbf{H}}^{BR}_1)^T\otimes \textbf{H}^{RU}_{1,1}\\ (\bar{\textbf{H}}^{BR}_1)^T\otimes \textbf{H}^{RU}_{2,1}\\ \vdots\\ (\bar{\textbf{H}}^{BR}_1)^T\otimes \textbf{H}^{RU}_{K,1} \\(\bar{\textbf{H}}^{BR}_2)^T\otimes \textbf{H}^{RU}_{1,2}\\ \vdots\\ (\bar{\textbf{H}}^{BR}_C)^T\otimes \textbf{H}^{RU}_{K,C}
\end{bmatrix}$ and $\textbf{h}_{\text{DL}}=\begin{bmatrix}
\text{vec}(-\bar{\textbf{H}}_{1,1})\\ \text{vec}(-\bar{\textbf{H}}_{2,1})\\ \vdots\\ \text{vec}(-\bar{\textbf{H}}_{K,1})\\ \text{vec}(-\bar{\textbf{H}}_{1,2})\\ \vdots\\ \text{vec}(-\bar{\textbf{H}}_{K,C})
\end{bmatrix}$.
\vspace{10pt}

The following theorem gives the number of antennas required at the relay for IBC-ODIA, where the proof is similar to that of IMAC-ODIA.

\vspace{10pt}
\begin{theorem}[Required number of relay antennas for the downlink cellular network with IBC-ODIA]
The relay beamformer $\textbf{T}_{\text{DL}}$ for the IBC-ODIA satisfying (\ref{equ_IA_cond_DL}) exists if $N_R\ge \max\{{(C-1)N,CKM}\}$.
\end{theorem}
\vspace{10pt}
With the proper design of $\textbf{T}_{\text{DL}}$, the inter-cell interferences are fully cancelled at the UEs. 
As noted above, intra-cell interferences will be aligned by designing the beamformer at the BS.
The beamformer design is described in the next subsection.

\subsection{Intra-Cell Interference Alignment by BS Beamforming}
The beamformer at the BS for the intra-cell IA can be designed jointly with the relay beamformer. 
For simplicity, each UE desires to receive $d$ independent data streams from its serving BS. 
The transmitted signal from $\text{BS}_j$ can be given as
\vspace{10pt}
\begin{equation*}
\textbf{x}_{j}=\textbf{V}_{(1,j)}\textbf{s}_{(1,j)}+\textbf{V}_{(2,j)}\textbf{s}_{(2,j)}+...+\textbf{V}_{(K,j)}\textbf{s}_{(K,j)},
\vspace{10pt}
\end{equation*}
where $\textbf{s}_{(k,j)}\in \mathbb K^{d\times 1}$ denotes the data stream for $\text{UE}_{(k,j)}$ from $\text{BS}_{j}$ and $\textbf{V}_{(k,j)}\in \mathbb K^{N\times d}$ is the beamforming matrix designed at $\text{BS}_{j}$ for $\text{UE}_{(k,j)}$.

The intra-cell interference at each UE is said to be fully aligned by the BS beamforming if the following condition is satisfied as

\setcounter{MaxMatrixCols}{3}
\begin{equation*}\label{equ_DL_beam}
\begin{bmatrix}
\textbf{H}_{(1,j),j}+\textbf{H}^{RU}_{(1,j)}\textbf{T}_{\text{DL}}\textbf{H}^{BR}_j\\ \vdots\\ \textbf{H}_{(K,j),j}+\textbf{H}^{RU}_{(K,j)}\textbf{T}_{\text{DL}}\textbf{H}^{BR}_j
\end{bmatrix}
\begin{bmatrix}
\textbf{V}_{(1,j)}&...&\textbf{V}_{(K,j)}
\end{bmatrix}=\textbf{A}_{j}\textbf{V}_{j}=
\begin{bmatrix}
\textbf{I}_{M\times d}&\textbf{0}&\textbf{0} \\  \textbf{0}&\ddots&\textbf{0} \\ \textbf{0}&\textbf{0}&\textbf{I}_{M\times d}
\end{bmatrix}
.
\vspace{10pt}
\end{equation*}
Since each entry of $\textbf{H}_{(k,j),j}$ is drawn i.i.d. from a continuous distribution, $\textbf{A}_{j}$ is either a full row or a full column rank matrix with probability one with a right or left inverse matrix.
 
Hence, $\textbf{V}_{i}$ is expressed as

\setcounter{MaxMatrixCols}{4}
\begin{equation}\label{equ_DL_beam2}
\textbf{V}_{j}= \textbf{A}_{j}^{\dagger}
\begin{bmatrix}
\textbf{I}_{M\times d}&\textbf{0}&\textbf{0} \\  \textbf{0}&\ddots&\textbf{0} \\ \textbf{0}&\textbf{0}&\textbf{I}_{M\times d}
\end{bmatrix}
=
\begin{bmatrix}
\{\textbf{A}_{j}^{\dagger}\}_{1:d}&\{\textbf{A}_{j}^{\dagger}\}_{M+1:M+d}&\cdots&\{\textbf{A}_{j}^{\dagger}\}_{(K-1)M+1:(K-1)M+d}
\end{bmatrix}
.
\vspace{10pt}
\end{equation}
From (\ref{equ_DL_beam2}), we can design $\textbf{V}_{j}$ with $\textbf{A}_{j}^{\dagger}$, where

\begin{equation*}
\textbf{A}_{j}^{\dagger}=
\begin{cases}
\textbf{A}_{j}^{\text{right}}=\textbf{A}_{j}^T(\textbf{A}_{j}\textbf{A}_{j}^T)^{-1}, \;{\rm if}\; N\ge KM \\
\textbf{A}_{j}^{\text{left}}=(\textbf{A}_{j}^T\textbf{A}_{j})^{-1}\textbf{A}_{j}^T, \;{\rm if}\; N< KM.
\end{cases}
\vspace{10pt}
\end{equation*}
Note that $N\ge Kd$ should be guaranteed for successful data transmission. 
Thus, $\textbf{V}_{(k,j)}$ makes it possible for $\text{BS}_j$ to deliver messages to $\text{UE}_{(k,j)}$ successfully with DoF $d$, where $d\le M$ should be satisfied.

After designing the relay beamformer and the BS beamformer, the inter-cell and intra-cell interferences can be completely removed at each UE. 
Accordingly, we summarize the achievable DoF for a downlink cellular network with IBC-ODIA via the following theorem. 

\vspace{10pt}
\begin{theorem}[Achievable DoF for a downlink cellular network with IBC-ODIA]\label{dof_DL}
The achievable DoF per cell and per UE for a $(C,K,M,N)$ downlink cellular network with IBC-ODIA is expressed as 
\begin{equation*}\label{equ_dof_DL}
\text{DoF}_{\text{DL,cell}}={\min\{Kd, N\}\over 2}\le {\min\{KM, N\}\over 2},
\vspace{10pt}
\end{equation*}
\begin{equation*}\label{equ_dof_DL2}
\text{DoF}_{\text{DL,UE}}={\text{DoF}_{\text{DL,cell}}\over K}.
\vspace{10pt}
\end{equation*}
\end{theorem}
Proof of \textit{Theorem} \ref{dof_DL} is similar to that of \textit{Theorem} \ref{the_dof_cell} and is thus omitted. 

\section{FD-ODIA for Full-Duplex Cellular Network} \label{sec_fdodia}

Furthermore, we focus on a cellular network, where full-duplex operation is feasible for BSs and UEs.
Due to the complexity of full-duplexing, it has not been a main concern for researchers, but it is known that full-duplexing, when possible, increases the throughput.
However, the need for a management scheme given the additional interference caused by full-duplex operation arises, which motivates the development of the proposed FD-ODIA scheme for a full-duplex cellular network.

\subsection{FD-ODIA for Full-Duplex Cellular Network}

Recently, several researchers have investigated interference alignment on full-duplex network due to the higher throughput achievable by the full-duplex mode \cite{fd16} \cite{fd15_1} \cite{fd15_2}.
In \cite{fd16}, it was proved that allowing for full-duplexing in each node exactly doubled the DoF performance compared to that of a half-duplexing multi-user interference channel.
The DoF derivation process in \cite{fd16} was based on network decomposition, where a similar methodology can be applied to a cellular network.
In \cite{fd15_1}, if each node operates in full-duplex mode, then a single full-duplex cell can be decomposed into two half-duplex cells, one on the uplink transmission and the other on the downlink transmission as shown in Fig. 6.

\begin{figure}
\centering
\includegraphics[scale=0.65]{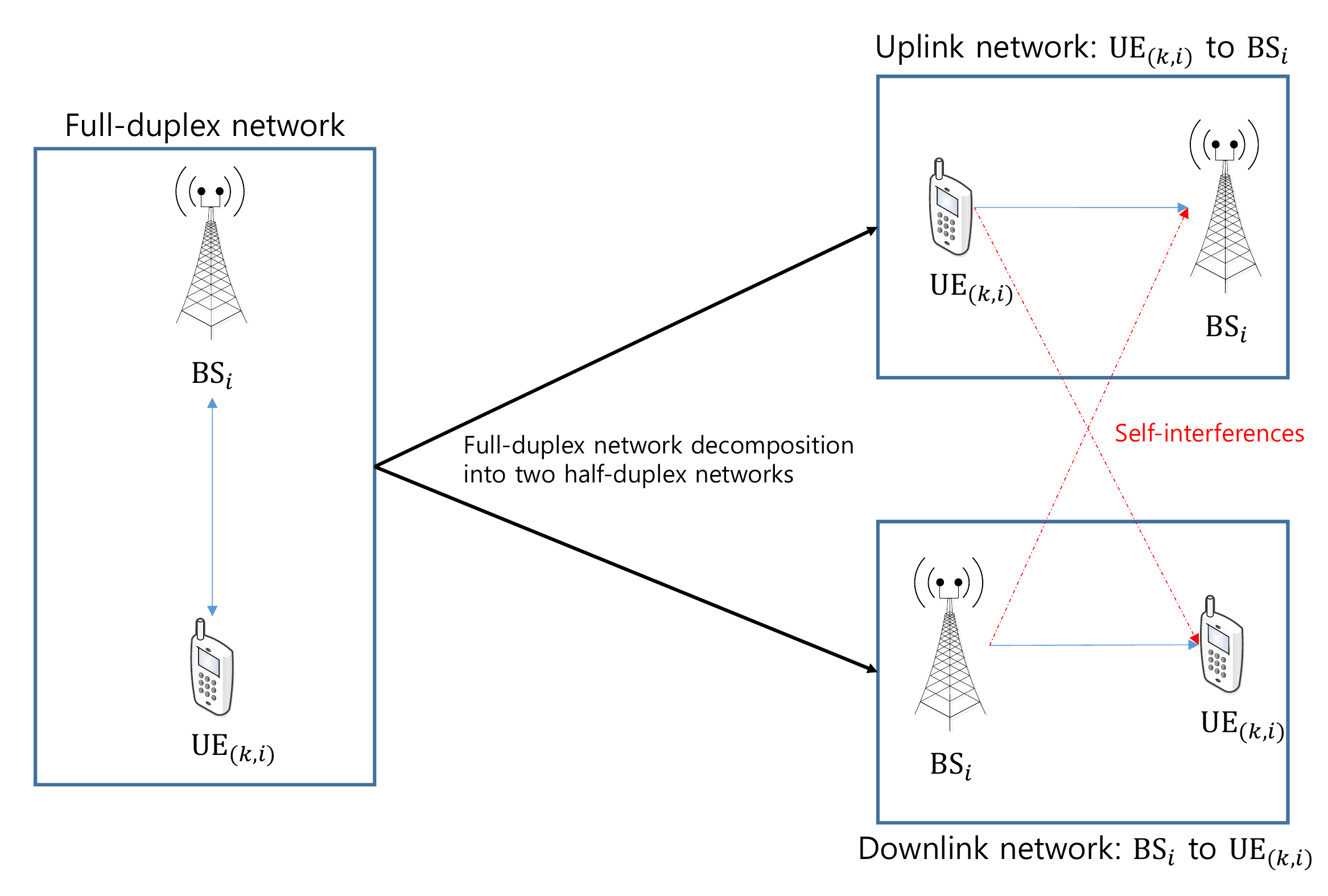}
\caption{Full-duplex network decomposition of a single full-duplex cell.}
\label{fig_network_decomposition}
\end{figure}

Considering full-duplex network decomposition as well as the fact that the proposed scheme can be operated on either the uplink or downlink, we can apply the IMAC-ODIA and the IBC-ODIA schemes proposed in the previous sections to a full-duplex cellular network, that is, FD-ODIA.
Although the relay operates in the half-duplex mode, the total DoF of the proposed FD-ODIA scheme can be doubled compared to the half-duplex case. 
We describe the symmetric full-duplex cellular network model below.

\begin{figure}
\centering
\includegraphics[scale=0.7]{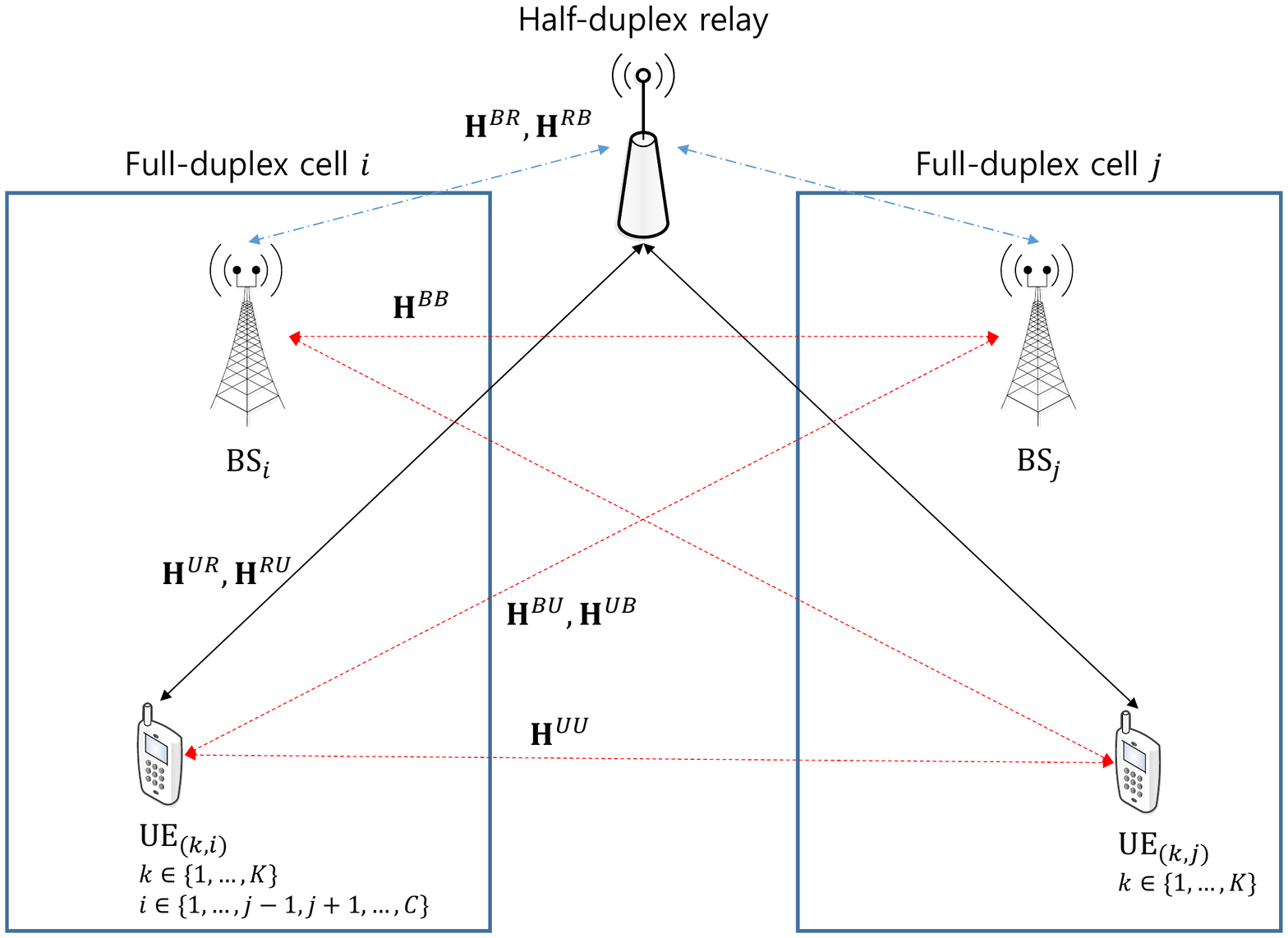}
\caption{Full-duplex cellular network served by a single half-duplex relay, where only interference-related channels are depicted.}
\label{fig_systemmodel_fd}
\end{figure}

Let $C$ and $K$ be the numbers of cells and users in each cell, respectively. 
Cell $j$ is served by a single base station in full-duplex mode denoted by $\text{BS}_j$ for $j\in \{1, ..., C\}$. 
User $(k,j)$ also operates in full-duplex mode denoted by $\text{UE}_{(k,j)}$, $k\in \{1, ..., K\}$, which is served by $\text{BS}_j$, as in Fig. \ref{fig_systemmodel_fd}. 
Note that the UEs only receive their desired signals from their corresponding BSs, implying the absence of device-to-device communication.
Each UE, each BS, and the relay is assumed to have $M$, $N$, and $N_{R}$ antennas, respectively. 
Note that the relay operates in half-duplex mode and the full-duplex relay will be discussed later.
Allowing for full-duplex operation at each source and destination node generates the following additional channel matrices.

- Channel from $\text{UE}_{(k,i)}$ to $\text{BS}_j$: $N\times M$ matrix $\textbf{H}^{UB}_{j,(k, i)}$ 

- Channel from $\text{UE}_{(k_1,i)}$ to $\text{UE}_{(k_2,j)}$: $M\times M$ matrix $\textbf{H}^{UU}_{(k_2,j),(k_1, i)}$ 

- Channel from $\text{UE}_{(k,j)}$ to relay: $N_R\times M$ matrix $\textbf{H}^{UR}_{(k, j)}$ 

- Channel from relay to $\text{UE}_{(k,j)}$: $M\times N_R$ matrix $\textbf{H}^{RU}_{(k,j)}$ 

- Channel from relay to $\text{BS}_j$: $N\times N_R$ matrix $\textbf{H}^{RB}_{j}$

- Channel from $\text{BS}_{i}$ to $\text{BS}_{j}$: $N\times N$ matrix $\textbf{H}^{BB}_{j,i}$

- Channel from $\text{BS}_{j}$ to $\text{UE}_{(k,i)}$: $M\times N$ matrix $\textbf{H}^{BU}_{(k,i),j}$

- Channel from $\text{BS}_{j}$ to relay: $N_R\times N$ matrix $\textbf{H}^{BR}_{j}$

\subsection{Achievable DoF of Full-Duplex Cellular Network with FD-ODIA}

Here, we derive the achievable DoF and the number of antennas required at the relay for a full-duplex cellular network with FD-ODIA as in the following theorem. 

\vspace{10pt}
\begin{theorem}[Achievable DoF for a full-duplex cellular network with FD-ODIA]\label{dof_FD}
For $N_R\ge C(KM+N)$, the achievable DoFs per cell, per BS, and per UE of the $(C,K,M,N)$ full-duplex cellular network are given as
\begin{equation*}\label{equ_dof_FD}
\text{DoF}_{\text{FD,cell}}=\text{DoF}_{\text{FD,BS}}+K\times \text{DoF}_{\text{FD,UE}}=\min\{KM,N\},
\vspace{10pt}
\end{equation*}
\begin{equation*}\label{equ_dof_FD_BS}
\text{DoF}_{\text{FD,BS}}={\min\{KM,N\}\over 2},
\vspace{10pt}
\end{equation*}
\begin{equation*}\label{equ_dof_FD_UE}
\text{DoF}_{\text{FD,UE}}={\min\{M,{N\over K}\}\over 2}
.
\vspace{10pt}
\end{equation*}
\end{theorem}

\begin{proof}
We prove this theorem using the IMAC-ODIA and IBC-ODIA schemes.
As noted earlier, the full-duplex cellular network with a relay can be decomposed into disjoint uplink and downlink networks sharing the relay as in Fig. \ref{fig_systemmodel_fd}.
Note that additional interference occurs due to full-duplexing, that is, UE to UE channels in the intra-cells or inter-cells now become interfering channels.

IBC-ODIA requires additional beamforming at the BS, in contrast to IMAC-ODIA.
Thus, it is natural that the FD-ODIA scheme should include beamforming at the BS. 
To describe the FD-ODIA scheme, we start with matrix augmentation with the same approach used earlier, where the relay only aligns the inter-cell interferences as 

\begin{align*}
\bar{\textbf{H}}^{BS}_j=\Big[&
\textbf{H}^{BB}_{j,1}\text{ } ...\text{ } \textbf{H}^{BB}_{j,j-1}\text{ } \textbf{H}^{BB}_{j,j+1}\text{ } ...\text{ } \textbf{H}^{BB}_{j,C} 
\\
&\textbf{H}^{UB}_{j,(1,1)}\text{ } ...\text{ } \textbf{H}^{UB}_{j,(K,j-1)}\text{ } \textbf{H}^{UB}_{j,(1,j+1)}\text{ } ...\text{ } \textbf{H}^{UB}_{j,(K,C)}
\Big]
\in \mathbb K^{N \times (C-1)(KM+N)}
,
\\ 
\\
\bar{\textbf{H}}^{UE}_{(k,j)}=\Big[&
\textbf{H}^{BU}_{(k,j),1}\text{ } ...\text{ } \textbf{H}^{BU}_{(k,j),j-1}\text{ } \textbf{H}^{BU}_{(k,j),j+1}\text{ } ...\text{ } \textbf{H}^{BU}_{(k,j),C} 
\\
&\textbf{H}^{UU}_{(k,j),(1,1)}\text{ } ...\text{ } \textbf{H}^{UU}_{(k,j),(k-1,j)}\text{ } \textbf{H}^{UU}_{(k,j),(k+1,j)}\text{ } ...\text{ } \textbf{H}^{UU}_{(k,j),(K,C)}
\Big]
\in \mathbb K^{M \times \{(CK-1)M+(C-1)N\}},
\\
\\
\bar{\textbf{H}}^{BR}_{j}=\Big[&
\textbf{H}^{BR}_{1}\text{ } ...\text{ } \textbf{H}^{BR}_{j-1}\text{ } \textbf{H}^{BR}_{j+1}\text{ } ...\text{ } \textbf{H}^{BR}_{C}
\\ 
&\textbf{H}^{UR}_{(1,1)}\text{ } ...\text{ } \textbf{H}^{UR}_{(K,j-1)}\text{ } \textbf{H}^{UR}_{(1,j+1)}\text{ } ...\text{ } \textbf{H}^{UR}_{(K,C)}
\Big]
\in \mathbb K^{N_R \times (C-1)(KM+N)},
\\
\\
\bar{\textbf{H}}^{UR}_{(k,j)}=\Big[&
\textbf{H}^{BR}_{1}\text{ } ...\text{ } \textbf{H}^{BR}_{j-1}\text{ } \textbf{H}^{BR}_{j+1}\text{ } ...\text{ } \textbf{H}^{BR}_{C}
\\ 
&\textbf{H}^{UR}_{(1,1)}\text{ } ...\text{ } \textbf{H}^{UR}_{(k-1,j)}\text{ } \textbf{H}^{UR}_{(k+1,j)}\text{ } ...\text{ } \textbf{H}^{UR}_{(K,C)}
\Big]
\in \mathbb K^{N_R \times \{(CK-1)M+(C-1)N\}}
.
\vspace{10pt}
\end{align*}
\\
In addition, the interference vectors should be redefined for the additional interferences as

\setcounter{MaxMatrixCols}{1}
\begin{equation*}
\bar{\textbf{x}}^{BS}_j=
\begin{bmatrix}
\textbf{x}^{BS}_{1}\\ \vdots\\ \textbf{x}^{BS}_{j-1}\\ \textbf{x}^{BS}_{j+1}\\ \vdots\\ \textbf{x}^{BS}_{C}\\
\textbf{x}^{UE}_{(1,1)}\\ \vdots\\ \textbf{x}^{UE}_{(K,j-1)}\\ \textbf{x}^{UE}_{(1,j+1)}\\ \vdots\\ \textbf{x}^{UE}_{(K,C)}
\end{bmatrix}\in \mathbb K^{(C-1)(KM+N) \times 1}
, \hspace{25pt}
\bar{\textbf{x}}^{UE}_{(k,j)}=
\begin{bmatrix}
\textbf{x}^{BS}_{1}\\ \vdots\\ \textbf{x}^{BS}_{j-1}\\ \textbf{x}^{BS}_{j+1}\\ \vdots\\ \textbf{x}^{BS}_{C}\\
\textbf{x}^{UE}_{(1,1)}\\ \vdots\\ \textbf{x}^{UE}_{(k-1,j)}\\ \textbf{x}^{UE}_{(k+1,j)}\\ \vdots\\ \textbf{x}^{UE}_{(K,C)}
\end{bmatrix}\in \mathbb K^{\{(CK-1)M+(C-1)N\} \times 1}
.
\vspace{10pt}
\end{equation*}
Note that $\textbf{x}^{BS}_{j}$ and $\textbf{x}^{UE}_{(k,j)}$ are the transmitted vectors from $\text{BS}_{j}$ and $\text{UE}_{(k,j)}$, respectively.
The definitions of the desired signal channel set $\{\widehat{\textbf{H}}\}$ and the desired signal vector set $\{\widehat{\textbf{x}}\}$ are exactly the same as those of the half-duplex uplink and downlink cellular network cases.

As noted above, the augmented interference matrices only contain the inter-cell interferences. 
Thus, additional BS beamforming will align the intra-cell interferences at the UEs, which will be described later.

The IA condition for FD-ODIA is similar to the IMAC/IBC-ODIA case except for the fact that it should be satisfied both on the BSs and the UEs as 
\vspace{10pt}
\begin{align}\label{equ_IA_FD}
\bar{\textbf{H}}^{BS}_j+\textbf{H}^{RB}_{j}\textbf{T}_{\text{FD}}\bar{\textbf{H}}^{BR}_{j}&=0,\text{ } j\in \{1, ..., C\},
\\
\label{equ_IA_FD2}
\bar{\textbf{H}}^{UE}_{(k,i)}+\textbf{H}^{RU}_{(k,i)}\textbf{T}_{\text{FD}}\bar{\textbf{H}}^{UR}_{(k,i)}&=0,\text{ } k\in \{1, ..., K\},\text{ } i\in \{1, ..., C\}
.
\end{align}
The existence of the relay beamformer $\textbf{T}_{\text{FD}}$ for the FD-ODIA scheme is also proved in a manner similar to that of the IMAC/IBC-ODIA case except for the number of relay antennas. 
In fact, in the full-duplex network case, the number of relay antennas $N_R$ should satisfy
\vspace{10pt}
\begin{equation*}
N_R\ge \max\{{(C-1)(KM+N),(C-1)N+(CK-1)M,C(KM+N)}\}=C(KM+N)
,
\vspace{10pt}
\end{equation*}
where $C(KM+N)$ is the total number of antennas in the cellular network.

Hence, $\textbf{T}_{\text{FD}}$ satisfying (\ref{equ_IA_FD}) and (\ref{equ_IA_FD2}) is derived as 

\setcounter{MaxMatrixCols}{1}
\begin{equation*} \label{equ_relay_beam_FD}
\textbf{T}_{\text{FD}}=\text{vec}_{N_R}^{-1}\left(
\begin{bmatrix}
(\bar{\textbf{H}}^{BR}_1)^T\otimes \textbf{H}^{RB}_{1,1}\\ \vdots\\ (\bar{\textbf{H}}^{BR}_C)^T\otimes \textbf{H}^{RB}_{C} \\(\bar{\textbf{H}}^{UR}_{(1,1)})^T\otimes \textbf{H}^{RU}_{(1,1)}\\ \vdots\\ (\bar{\textbf{H}}^{UR}_{(K,C)})^T\otimes \textbf{H}^{RU}_{(K,C)}
\end{bmatrix}
^{\text{right}}
\begin{bmatrix}
\text{vec}(-\bar{\textbf{H}}^{BS}_{1})\\ \vdots\\ \text{vec}(-\bar{\textbf{H}}^{BS}_{C})\\ \text{vec}(-\bar{\textbf{H}}^{UE}_{(1,1)})\\ \vdots\\ \text{vec}(-\bar{\textbf{H}}^{UE}_{(K,C)})
\end{bmatrix}
\right).
\vspace{10pt}
\end{equation*}
Note that the right null vector term in the argument of $\text{vec}_{N_R}^{-1}(\cdot)$ is omitted.

Furthermore, it is necessary to design the beamformer $\textbf{V}^{\text{FD}}_j$ at $\text{BS}_j$ in the same manner as IBC-ODIA, which should satisfy the following intra-cell IA condition 

\setcounter{MaxMatrixCols}{3}
\begin{equation*}\label{equ_FD_beam}
\begin{bmatrix}
\textbf{H}^{BU}_{(1,j),j}+\textbf{H}^{RU}_{(1,j)}\textbf{T}_{\text{FD}}\textbf{H}^{BR}_j\\ \vdots\\ \textbf{H}^{BU}_{(K,j),j}+\textbf{H}^{RU}_{(K,j)}\textbf{T}_{\text{FD}}\textbf{H}^{BR}_j
\end{bmatrix}
\begin{bmatrix}
\textbf{V}^{\text{FD}}_{(1,j)}&...&\textbf{V}^{\text{FD}}_{(K,j)}
\end{bmatrix}=\textbf{A}^{\text{FD}}_{j}\textbf{V}^{\text{FD}}_{j}=
\begin{bmatrix}
\textbf{I}_{M\times d}&\textbf{0}&\textbf{0} \\  \textbf{0}&\ddots&\textbf{0} \\ \textbf{0}&\textbf{0}&\textbf{I}_{M\times d}
\end{bmatrix}
.
\vspace{10pt}
\end{equation*}
Similar to the IBC-ODIA case, $\textbf{A}_{j}^{\text{FD}}$ is either a full row or a full column rank matrix with probability one depending on the number of antennas.
In this case, $\textbf{V}^{\text{FD}}_j$ can be derived as

\setcounter{MaxMatrixCols}{4}
\begin{equation*}\label{equ_FD_beam2}
\textbf{V}^{\text{FD}}_{j}=
\begin{bmatrix}
\{(\textbf{A}_{j}^{\text{FD}})^{\dagger}\}_{1:d}&\{(\textbf{A}_{j}^{\text{FD}})^{\dagger}\}_{M+1:M+d}&\cdots&\{(\textbf{A}_{j}^{\text{FD}})^{\dagger}\}_{(K-1)M+1:(K-1)M+d}
\end{bmatrix},
\vspace{10pt}
\end{equation*}
where $(\textbf{A}_{j}^{\text{FD}})^{\dagger}$ is either the right or left inverse matrix of $\textbf{A}_{j}^{\text{FD}}$.

Since the interferences are fully aligned and cancelled at each BS, the DoF of $\text{BS}_j$ for the full-duplex cellular network can be achieved as

\begin{equation*}
\text{DoF}^{\text{FD,BS}}_{j}={\min\{KM,N\}\over 2}
,
\vspace{10pt}
\end{equation*}
and the DoF of $\text{UE}_{(k,j)}$ can also be achieved as

\begin{equation*}
\text{DoF}^{\text{FD,UE}}_{(k,j)}={\min\{KM,N\}\over 2K}={\min\{M,{N\over K}\}\over 2}
.
\vspace{10pt}
\end{equation*}
Thus, we prove the theorem.
\end{proof}

Note that the total achievable DoF of the proposed FD-ODIA is doubled compared to the half-duplex case. 
It is also interesting that the proposed scheme requires only a half-duplexing AF relay to achieve full-duplexing gain in a full-duplex cellular network with FD-ODIA.
The design methodology of the proposed IMAC/IBC-ODIA schemes can directly be applied to the FD-ODIA scheme.

In the previous work \cite{fd16}, it was shown that causal MIMO full-duplex relay based on the DF protocol cannot increase the DoF but that non-causal and instantaneous full-duplex DF relaying can increase the DoF.
The result in \textit{Theorem \ref{dof_FD}} can be directly applied to the case of full-duplex instantaneous AF relay as in the following corollary.

\vspace{10pt}
\begin{corollary}[Achievable DoF with full-duplex instantaneous AF relay]\label{dof_FD_inst}
The achievable DoFs per cell, BS, and UE of the same full-duplex cellular network with full-duplex instantaneous AF relaying are given as
\begin{equation*}\label{equ_dof_FD_inst}
\text{DoF}_{\text{FD,cell}}^{\text{inst}}=\text{DoF}_{\text{FD,BS}}^{\text{inst}}+K\times \text{DoF}_{\text{FD,UE}}^{\text{inst}}=2\min\{KM,N\},
\vspace{10pt}
\end{equation*}
\begin{equation*}\label{equ_dof_FD_BS_inst}
\text{DoF}_{\text{FD,BS}}^{\text{inst}}=\min\{KM,N\},
\vspace{10pt}
\end{equation*}
\begin{equation*}\label{equ_dof_FD_UE_inst}
\text{DoF}_{\text{FD,UE}}^{\text{inst}}=\min\{M,{N\over K}\}
,
\vspace{10pt}
\end{equation*}
\end{corollary}
where the relay antenna requirement is the same as that in \textit{Theorem \ref{dof_FD}}. 

\begin{remark}[Half-duplex UE with full-duplex BS]
Due to the hardware complexity limitation, an UE cannot be operated in full-duplex mode, and thus half-duplex UEs served by full-duplex BSs can be considered.
In this case, the UEs can be partitioned into either an uplink group or a downlink group. 
Assume that $K_1$ UEs are in the uplink group and $K_2$ UEs are in the downlink group with $M$ antennas at each UE for both groups.
Note that the BSs operate in the full-duplex mode with $N$ antennas.
Then the DoF of the above system can be given as the following corollary.
\end{remark}

\begin{corollary}[DoF of a cellular network with half-duplex UEs served by a full-duplex BS]\label{coro_2}
The achievable DoF per cell with the half-duplex UEs served by the full-duplex BS is given as

\begin{equation}
\text{DoF}_{\text{HD+FD,cell}}=\text{DoF}_{\text{HD+FD,BS}}+K_2\times \text{DoF}_{\text{HD+FD,UE}} \label{equ_HD+FD},
\vspace{10pt}
\end{equation}
where 
$\text{DoF}_{\text{HD+FD,BS}}={\min\{{K_1}M,N\}\over 2}$ and
$\text{DoF}_{\text{HD+FD,UE}}={\min\{M,{N\over {K_2}}\}\over 2}$ .
\vspace{10pt}
\end{corollary}

The DoF in (\ref{equ_HD+FD}) can be achieved with the IA scheme similar to the previous FD-ODIA scheme.
It can be achieved by FD-ODIA with the exclusion of non-existing BS-to-UE channels, UE-to-BS channels, and UE-to-UE channels in the relay beamformer design, which vanishes automatically due to the half-duplexing of UEs.

Thus, it should be noted that the proposed schemes are very flexible, that is, IA schemes for various cellular networks can be implemented with the proposed schemes by modifying the channel assignment protocol in the relay.

\section{Discussion} \label{sec_discussion}

\subsection{DoF Improvement without Time Extension}

With the proposed schemes, the total achievable DoF of a half-duplex uplink cellular network $\text{DoF}_{\text{ODIA,uplink}}$ is given as $C\min\{N, KM\}\over 2$. 
For the $(C,K,M,N)$ uplink cellular networks, $\text{DoF}_{\text{ODIA,uplink}}$ of the proposed scheme with a relay and without UE beamforming is much larger than $\text{DoF}_{\text{linear,uplink}}$ of the linear scheme without a relay and with UE beamforming, which can be analyzed using the outer bound of DoF of the cellular network without time extension.

\begin{figure}
\centering
\includegraphics[scale=0.55]{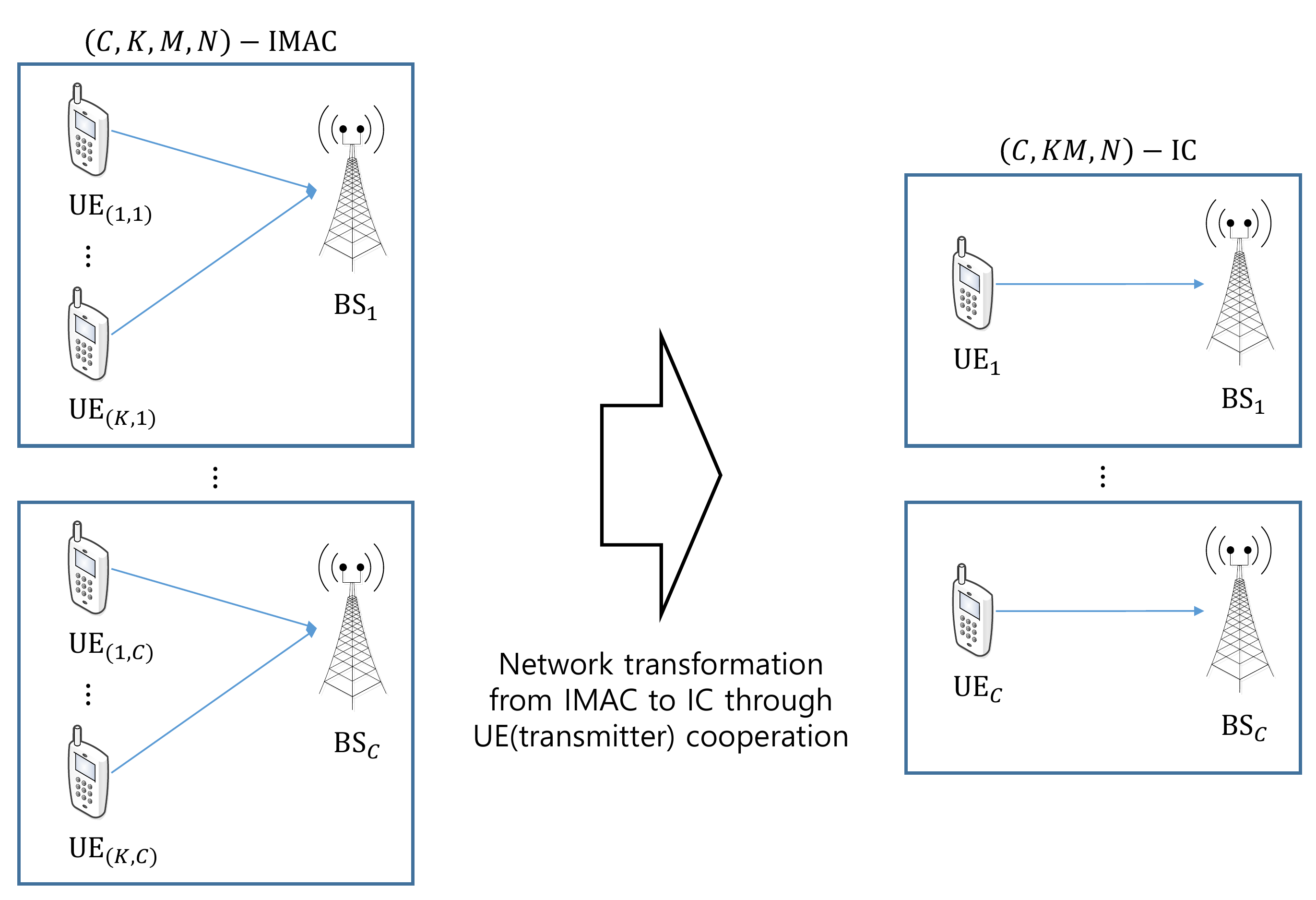}
\caption{IMAC-to-IC transformation with $\text{UE}_{c},\text{ } 1\le c\le C$, which is a super-UE with $KM$ antennas.}
\label{fig_network_transformation}
\end{figure}

An $(C,K,M,N)$ IMAC can be transformed into an $(C,KM,N)$ IC with transmitter cooperation, that is, all UEs at each BS transmit their data jointly as a single super-UE as shown in Fig. \ref{fig_network_transformation}.
Let $\text{DoF}_{\text{Linear,coop}}$ denote the total DoF of the $(C,KM,N)$ IC.
Then, the following inequality is easily given as
\begin{equation*}\label{inequ_dof_transform}
\text{DoF}_{\text{linear,uplink}} \le \text{DoF}_{\text{linear,coop}}
.
\vspace{10pt}
\end{equation*}
%
%
Further, using the result in \cite{linear12}, $\text{DoF}_{\text{linear,coop}}$ is bounded as

\begin{equation}\label{inequ_dof_transform_coop}
\text{DoF}_{\text{linear,uplink}} \le \text{DoF}_{\text{linear,coop}} \le KM+N
.
\vspace{10pt}
\end{equation}
However, the total DoF $\text{DoF}_{\text{ODIA,uplink}}$ of the proposed scheme increases proportional to the number of cells, which cannot be achieved in (\ref{inequ_dof_transform_coop}).
From \cite{relay06}, it is easy to find that $\text{DoF}_{\text{ODIA,uplink}}={C\min\{N, KM\}\over 2}$ is indeed the maximally achievable DoF in an uplink cellular network with a half-duplex relay.

Note that the information-theoretic upper bound of the $(C,K,M,N)$ IMAC ($\approx {KMN\over {KM+N}}$, see \cite{cellular15}) is approximately two times larger than the achievable DoF of the proposed scheme in an extreme case (i.e., $N\ll KM$ or $N\gg KM$), 
whereas the upper bound cannot be achieved by the linear IA scheme without time extension.

Additionally, in an antenna regime where BSs and UEs in a cell have similar numbers of antennas (i.e., $N\approx KM$), the proposed IMAC/IBC-ODIA schemes achieve the optimal DoF of MIMO IMAC/IBC, which is believed to be achieved only through asymptotic IA or IA aided by a full-duplex relay.

\subsection{Possibility of Gain from Additional Relay Antennas}

Suppose that a relay antenna setup is sufficient to serve an entire network, that is, there are more antennas than required. 
Then the question is, ``Is there any extra gain from additional relay antennas?"  
We answer this question in this subsection.

Consider that an uplink cellular network with the IA-applied data stream through decorrelator at the BS is expressed as (\ref{equ_deco}) in the previous section.
Our goal is to design a relay beamformer satisfying the following additional statement as

\begin{equation}\label{equ_boost}
\widehat{\textbf{H}}_{\text{eff},j}=\widehat{\textbf{H}}_j+\textbf{H}^{RB}_{j}\textbf{T}_{+}\widehat{\textbf{H}}^{UR}_j=(1+\alpha)\widehat{\textbf{H}}_j
\hspace{10pt}
\Rightarrow
\hspace{10pt}
\textbf{H}^{RB}_{j}\textbf{T}_{+}\widehat{\textbf{H}}^{UR}_j=\alpha\widehat{\textbf{H}}_j,
\vspace{10pt}
\end{equation}
where $\alpha \ge 0$ and $\textbf{T}_{+}$ denotes the relay beamformer with additional gain.
If (\ref{equ_boost}) is satisfied, it can be interpreted as boosting the received signal power level by $\alpha$, while the DoF performance remains the same. 

We skip the details because the proof is identical to those in Theorems 1 and 2. 
The existence of $\textbf{T}_{+}$ satisfying  (\ref{equ_boost}) can be guaranteed when $N_R\ge \max\{{CKM,CN}\}$. 
Hence, the design of $\textbf{T}_{+}$ is given as
\setcounter{MaxMatrixCols}{1}
\begin{equation*} \label{equ_relay_beam_add}
\textbf{T}_{+}=\text{vec}_{N_R}^{-1}\left(
\begin{bmatrix}
({\textbf{H}}^{UR}_1)^T\otimes \textbf{H}^{RB}_{1}\\ ({\textbf{H}}^{UR}_2)^T\otimes \textbf{H}^{RB}_{2}\\ \vdots\\ ({\textbf{H}}^{UR}_C)^T\otimes \textbf{H}^{RB}_{C}
\end{bmatrix}
^{\text{right}}
\begin{bmatrix}
\text{vec}({\textbf{H}}_1)\\ \text{vec}({\textbf{H}}_2)\\ \vdots\\ \text{vec}({\textbf{H}}_C)
\end{bmatrix}
\right),
\vspace{10pt}
\end{equation*}
where for simplicity, the right null vector term in the argument of $\text{vec}_{N_R}^{-1}(\cdot)$ is omitted and ${\textbf{H}}_j\in \mathbb K^{N \times CKM}$ and ${\textbf{H}}^{UR}_j\in \mathbb K^{N_R \times CKM}$ are correspondingly defined as

\setcounter{MaxMatrixCols}{9}
\begin{align}
{\textbf{H}}^{UR}_j&=
\begin{bmatrix}
\textbf{H}^{UR}_{1, 1}&\textbf{H}^{UR}_{2, 1}&\cdots&\textbf{H}^{UR}_{K, C},
\end{bmatrix}, \nonumber
\\
\nonumber
\\
{\textbf{H}}_j&=
\begin{bmatrix}
-\textbf{H}_{j,(1,1)}&\cdots&-\textbf{H}_{j,(K,j-1)}&\alpha\textbf{H}_{j,(1,j)}&\cdots&\alpha\textbf{H}_{j,(K,j)}&-\textbf{H}_{1,(1,j+1)}&\cdots&-\textbf{H}_{j,(K,C)}
\end{bmatrix}
. \nonumber
\end{align}
\\
Then, the relay beamformer $\textbf{T}_{+}$ for the proposed IMAC-ODIA can align the interferences, while also boosting the desired signal power level to $KM$ additional antennas at the relay.
Thus, from the additional relay antenna, we can achieve better throughput.

\vspace{10pt}

\section{Conclusions} \label{sec_conclusion}

In this paper, we proposed IA schemes for the cellular networks which operate with a half-duplex relay without UE beamforming or CSI handling, representing the fundamental limit when applying IA to cellular network.
We proposed an uplink cellular network with IMAC-ODIA and derived its achievable DoF.
Using linear beamforming at the relay and uplink-downlink duality, we extended the proposed scheme with the BS beamformer design to a downlink cellular network without UE beamforming or CSI handling.
Further, the proposed schemes can also be extended to a full-duplex cellular network, achieving doubled DoF compared to a half-duplex cellular network.
Although it has DoF gap from the information-theoretic upper bound, we proved that the proposed schemes have network gain proportional to the network size, which cannot be achieved in a cellular network with linear IA. 
Moreover, for some antenna regime, the proposed schemes are able to achieve the optimal DoF.

\vspace{10pt}
\section*{Acknowledgement}
This work was supported by 

%








\end{document}